\documentclass[12pt,reqno,a4paper]{amsart}
\usepackage[left=3cm, top=3.5cm, bottom=3.5cm, right=3cm]{geometry}
\usepackage{mathrsfs}
\usepackage{amsmath,amsxtra,amsfonts,amssymb, amsthm, amscd, epsfig}
\usepackage[dvipsnames,usenames]{color}
\usepackage{amssymb,amsmath,amsfonts,epsfig}
\usepackage{graphicx}

\usepackage{tikz}
\usepackage{hyperref}

\numberwithin{equation}{section}

\newtheorem{thm}{Theorem}[section]

\newtheorem{cor}[thm]{Corollary}

\newtheorem{lem}[thm]{Lemma}

\newtheorem{rem}[thm]{Remark}

\newcommand{\eqa}{\begin{eqnarray}}
\newcommand{\eeqa}{\end{eqnarray}}
\newcommand{\beq}{\begin{equation}}
\newcommand{\eeq}{\end{equation}}

\newcommand{\p}{\partial}

 \def\res{\mathop{\text{\rm Res}}}

\def \dsum{\displaystyle\sum}

\allowdisplaybreaks[1]

\begin{document}

\title[]
{Principal Hierarchies for Frobenius Manifolds with Rational and Trigonometric Superpotentials}

\author[]{Shilin Ma}

\address[]{Shilin Ma,  School of Mathematics and Statistics,  Xi’an Jiaotong University, Xi’an 710049,  P. R.
	China,}
\email{mashilin@xjtu.edu.cn}
\date{\today}

\begin{abstract}
In this paper, we construct the principal hierarchies for Frobenius manifolds with rational and trigonometric superpotentials, as well as their almost dualities. We demonstrate that in both cases, submanifolds with even superpotentials form natural Frobenius submanifolds, and their principal hierarchies can be obtained as restrictions of the principal hierarchies for the original Frobenius manifolds. Furthermore, we introduce a natural rank-1 extension for each of these Frobenius manifolds, providing solutions to the associated open WDVV equations. The principal hierarchy for each extension is also explicitly constructed.

\vskip 2ex

\end{abstract}

\maketitle 

\tableofcontents
\section{Introduction}

The concept of Frobenius manifolds, first introduced by Dubrovin in \cite{dub1998}, provides a geometric framework for capturing the associativity equations inherent in two-dimensional topological field theory (2D TFT). This concept is highly relevant across various areas of mathematical physics, including Gromov-Witten theory, singularity theory, and integrable systems, among others. Its significance is further highlighted by applications in works such as \cite{manin1999frobenius, hertling2002frobenius, dubrovin2001normal, givental2005simple, liu2015bcfg, dubrovin2016hodge} and their references.

Every Frobenius manifold is associated with an integrable hierarchy of hydrodynamic type, referred to as the principal hierarchy for the Frobenius manifold.  This hierarchy involves unknown functions depending on a single scalar spatial variable and various time variables. In the case of semisimple Frobenius manifold, this hierarchy can be deformed into a dispersive hierarchy, referred to as the Dubrovin-Zhang hierarchy. The tau function, determined by the string equation of the Dubrovin-Zhang hierarchy, provides the partition function for the corresponding 2D TFT. This intricate relationship between Frobenius manifold and 2D TFT has significantly advanced our understanding of their geometric and algebraic structures. For further details, see \cite{dubrovin2001normal, liu2022variational, liu2023variational}.

Although the abstract theory of Dubrovin and Zhang is well-established, explicitly constructing the principal hierarchy for a given Frobenius manifold presents a certain level of difficulty.  For relatively low-dimensional Frobenius manifold, one can directly derive the principal hierarchy by solving the PDE systems that govern its Hamiltonian densities.  Examples of this process are provided in \cite{dubrovin2001normal, Tu_2003, Tu_2007}. For higher-dimensional cases, although constructions have been provided, as in \cite{aoyama1996topological, raimondo2012frobenius, carlet2015principal}, these rely on the specific structures of the corresponding Frobenius manifolds and lack general applicability.

The main result of this paper is the explicit construction of the principal hierarchies for Frobenius manifolds with rational and trigonometric  superpotentials, respectively. Our approach is as follows: we first provide suitable representations for the cotangent spaces and derive the explicit formulas for the Hamiltonian structures associated with the flat metrics for these Frobenius manifolds. Then, we reformulate the PDE systems governing the Hamiltonian densities of the principal hierarchies into algebraic equations involving the superpotentials, which admit straightforward solutions. Using a similar approach, we have also constructed the principal hierarchies for the almost dualities \cite{dubrovin2004almost} of these Frobenius manifolds.

Another important class of Frobenius manifolds consists of those with even rational and trigonometric superpotentials. We will show that these manifolds are natural Frobenius submanifolds \cite{STRACHAN2001, STRACHAN2004} of the two types of Frobenius manifolds discussed above, and that their associated principal hierarchies are direct restrictions of those of the parent manifolds.

The open WDVV equations were first introduced in \cite{horev2012open} in the context of open Gromov-Witten theory. P. Rossi observed that a solution to the open WDVV equations is equivalent to a flat F-manifold, which serves as a rank-1 extension of the given Frobenius manifold. This extension was systematically studied in \cite{alcolado2017extended}.  The descendant potential and Virasoro constraints for a flat F-manifold were constructed in \cite{basalaev2019open} for genus-zero cases and in \cite{arsie2023semisimple} for higher genera. The dispersive deformation of the principal hierarchy for a flat F-manifold, as a generalization of the DR hierarchy for a Frobenius manifold, was studied in \cite{arsie2021flat}. Examples related to the open Gromov-Witten theory of a point and open r-spin theory were investigated in \cite{pandharipande2014intersection, buryak2015equivalence, buryak2016open} and \cite{Buryak:2018ypm,buryak2024open}, respectively. In the present paper, we show that there exist natural rank-1 extensions for Frobenius manifolds with rational and trigonometric superpotentials, and explicitly construct the principal hierarchies for these extensions.

We note that the construction presented in this paper has a certain level of generality. For instance, in \cite{ma2024infinite}, we employed a similar approach to construct the principal hierarchy for the infinite-dimensional Frobenius manifold underlying the genus-zero universal Whitham hierarchy.

	Let us now state our main results precisely.
	
	Given positive integers $n_{0},\cdots,n_{m}$, let $M^{cKP}$ denote the space of rational functions of the form:
	\begin{equation}\label{kdvsup}
		\lambda(z)=\frac{1}{n_{0}}z^{n_{0}}+a_{0,n_{0}-2}z^{n_{0}-2}+\cdots+a_{0,0}+\sum_{i=1}^{m}\sum_{j=1}^{n_{i}}a_{i,j}(z-a_{i,0})^{-j},
	\end{equation}
	where the set of coefficients $\{a_{0,i}\}_{i=0}^{n_{0}-2}\cup\{a_{1,i}\}_{i=0}^{n_{1}}\cup\cdots\cup \{a_{m,i}\}_{i=0}^{n_{m}}$ serves as local coordinates on the manifold $M^{cKP}$.
	The space $M^{cKP}$
	can be regarded as a special type of Hurwitz space and is thus equipped with a Frobenius manifold structure constructed by Dubrovin \cite{dub1998}.  When \( m=0 \), this structure coincides with that on the orbit space of the Coxeter group  of type \( A \) \cite{dubrovin1998geometry}.
	
	The flat coordinates on \( M^{cKP} \), denoted as \( \mathbf{t}^{cKP} = \{t_{0,j}\}_{j=1}^{n_0-1} \cup \{t_{1,j}\}_{j=0}^{n_1} \cup \cdots \cup \{t_{m,j}\}_{j=0}^{n_m} \), are given by the following expansions:
	\[
	z = 
	\begin{cases}
		t_{i,0} + t_{i,1}w_{i}(z)^{-1} + \cdots, & \text{as } z \to a_{i,0},\ i = 1, \cdots, m, \\
		w_{0}(z) - t_{0,1}w_{0}(z)^{-1} - t_{0,2}w_{0}(z)^{-2} + \cdots, & \text{as } z \to \infty,\ i = 0,
	\end{cases}
	\]
	where
	\[
	w_{i}(z) = 
	\begin{cases}
		(n_i\lambda(z))^{\frac{1}{n_i}} = w_{i,1}(z-a_{i,0})^{-1} + \cdots, & \text{as } z \to a_{i,0},\ i = 1, \cdots, m, \\
		(n_0\lambda(z))^{\frac{1}{n_0}} = z + w_{0,1}z^{-1} + \cdots, & \text{as } z \to \infty,\ i = 0.
	\end{cases}
	\]

\begin{thm}\label{maincKP}
	The Hamiltonian densities of the principal hierarchy for the Frobenius manifold \( M^{cKP} \) are given by
	\begin{align*}
		\theta_{t_{0,j},p}(z) =& -\res_{\infty} c_{0,j;p} w_{0}(z)^{(p+1)n_{0}-j}\, dz, \quad j = 1, \ldots, n_{0} - 1, \\
		\theta_{t_{i,j},p}(z) =& \res_{a_{i,0}} c_{i,j;p} w_{i}(z)^{(p+1)n_{i}-j} \, dz, \quad i = 1, \ldots, m, \ j = 0, \ldots, n_{i} - 1,
	\end{align*}
	and
	\begin{align*}
			\theta_{t_{i,n_{i}},p}(z) =& -\res_{\infty} \frac{\lambda(z)^{p}}{p!} (\log \frac{w_{0}(z)}{z - a_{i,0}} - \frac{c_{p}}{n_{0}}) \, dz + \res_{a_{i,0}} \frac{\lambda(z)^{p}}{p!} (\log(z - a_{i,0}) w_{i}(z) - \frac{c_{p}}{n_{i}}) \, dz \\
		& + \sum_{s \neq i} \res_{a_{s,0}} \frac{\lambda(z)^{p}}{p!} \log(z - a_{i,0}) \, dz, \quad i = 1, \ldots, m,
	\end{align*}
	where the constants $\{c_{i,j;p}\}$ and $\{c_{p}\}$ are defined as
	\[
	c_{i,j;p} = \frac{1}{n_{i} - j} \frac{1}{2n_{i} - j} \cdots \frac{1}{(p + 1)n_{i} - j}, \quad c_{p} = \sum_{s = 1}^{p} \frac{1}{s}.
	\]
	The corresponding Hamiltonian vector fields \( \frac{\partial}{\partial T^{\bullet,p}} = \mathcal{P}(d\theta_{\bullet,p+1}) \) are given by 
	\begin{align*}
		\frac{\partial \lambda(z)}{\partial T^{t_{0,j},p-1}} =& \{(c_{0,j;p-1} w_{0}(z)^{p n_{0} - j})_{\infty, \geq 0}, \lambda(z)\}, \quad j = 1, \ldots, n_{0} - 1, \\
		\frac{\partial \lambda(z)}{\partial T^{t_{i,j},p-1}} =& -\{(c_{i,j;p-1} w_{i}(z)^{p n_{i} - j})_{a_{i,0}, \leq -1}, \lambda(z)\}, \quad i = 1, \ldots, m, \ j = 0, \ldots, n_{i} - 1,
	\end{align*}
	and
	\begin{align*}
			\frac{\partial \lambda(z)}{\partial T^{t_{i,n_{i}},p-1}} =& \{(\frac{\lambda(z)^{p}}{p!} (\log \frac{w_{0}(z)}{z - a_{i,0}} - \frac{c_{p}}{n_{0}}))_{\infty, \geq 0}-(\frac{\lambda(z)^{p}}{p!} (\log(z - a_{i,0}) w_{i}(z) - \frac{c_{p}}{n_{i}}))_{a_{i,0}, \leq -1}, \lambda(z)\}\\
		& - \sum_{s \neq i} \{(\frac{\lambda(z)^{p}}{p!} \log(z - a_{i,0}))_{a_{s,0}, \leq -1}, \lambda(z)\} + \{\frac{\lambda(z)^{p}}{p!} \log(z - a_{i,0}), \lambda(z)\},
	\end{align*}
	where the Poisson bracket is defined as
	\[
	\{f(z,x), g(z,x)\} = \frac{\partial f(z,x)}{\partial z} \frac{\partial g(z,x)}{\partial x} - \frac{\partial f(z,x)}{\partial x} \frac{\partial g(z,x)}{\partial z}.\]
\end{thm}

When \( m=1 \), this principal hierarchy was constructed by Aoyama and Kodama using a different method \cite{aoyama1996topological}, and serves as an extension of the dispersionless limit of the constrained KP hierarchy \cite{liu2015central}. The corresponding Dubrovin-Zhang hierarchy governs the generating function enumerating rooted hypermaps on compact two-dimensional surfaces \cite{carlet2023enumeration}.

A similar approach can be used to construct the principal hierarchy for the almost duality of $M^{cKP}$.

\begin{thm}\label{maincKPdul}
	Let $\hat{M}^{cKP}$ be the almost duality of the Frobenius manifold $M^{cKP}$, then the Hamiltonian densities of the principal hierarchy for $\hat{M}^{cKP}$ are given by
	$$
	F_{\tilde{\gamma},p} = \frac{1}{2\pi\mathrm{i}} \int_{\tilde{\gamma}} \frac{(\log(\lambda(z)))^{p+1}}{(p+1)!} \, dz,
	$$
	where $\tilde{\gamma}$ is a simple closed curve in the complex plane $\mathbb{C}$ such that the winding number of $\lambda(z)$ along $\tilde{\gamma}$ is zero. 
	
	In particular, let $q_{1}, \ldots, q_{k}$ and $p_{1}, \ldots, p_{k}$ be the zeros and poles of $\lambda(z)$, respectively, within the region enclosed by $\tilde{\gamma}$. Then
	$$
	F_{\tilde{\gamma},0} = \sum_{s=1}^{k} p_{s} - \sum_{s=1}^{k} q_{s}.
	$$
\end{thm}

In the case of \( m = 0 \), the Hamiltonian densities of this principal hierarchy were constructed by Dubrovin using period integrals \cite{dubrovin2004almost}.

As a direct consequence of Theorem \ref{maincKP}, we obtain the explicit form of the principal hierarchy for the Frobenius submanifold \cite{STRACHAN2001} of $M^{cKP}$ with even superpotential. Assume there exist positive integers $m'$ and $n_{0}',n_{1}',\cdots,n_{m'}'$ such that
$$
	m=2m'-1,\quad n_{0}=2n_{0}',\quad n_{1}=2n_{1}',\quad n_{2j-2}= n_{2j-1}=n_{j}',\quad 2\le j\le m'.
$$
Let $M^{D-cKP}$ be a submanifold of $M^{cKP}$ consisting of elements of the form:
\begin{equation}\label{Dckpsup}
	\lambda(z)=\frac{1}{2n_{0}'}z^{2n_{0}'}+\sum_{j=0}^{n_{0}'-1}b_{0,j}z^{2j}+\sum_{j=1}^{n_{1}'}b_{1,j}z^{-2j}+\sum_{i=2}^{m'}\sum_{j=1}^{n_{i}'}b_{i,j}(z^2-b_{i,0})^{-j},
\end{equation}
which is characterized by the condition $\lambda(z)=\lambda(-z)$.

\begin{thm}\label{maincKPsym}
	\( M^{D-cKP} \) is a natural Frobenius submanifold of \( M^{cKP} \), characterized by the following conditions imposed on the flat coordinates of \( M^{cKP} \): 
	\begin{align*}
		&t_{0,2} = t_{0,4} = \cdots = t_{0,2n_{0}'-2} = 0; \\
		&t_{1,0} = t_{1,2} = \cdots = t_{1,2n_{1}'} = 0; \\
		&t_{2i-2,j} = -t_{2i-1,j}, \quad  2 \leq i \leq m', \ 0 \leq j \leq n_{i}'.
	\end{align*}
	Moreover, the subhierarchy \( \frac{\partial}{\partial T^{t,p}} \) of the principal hierarchy for \( M^{cKP} \), where
	\[
	t \in \textbf{t}^{D-cKP} = \{t_{0,2j-1}\}_{j=1}^{n_{0}'} \cup \{t_{1,2j-1}\}_{j=1}^{n_{1}'} \cup \{t_{2,j}\}_{j=0}^{n_{2}'} \cup \{t_{4,j}\}_{j=0}^{n_{3}'} \cup \cdots \cup \{t_{2m'-2,j}\}_{j=0}^{n_{m'}'},
	\]
	can be directly restricted to the submanifold \( M^{D-cKP} \), forming the principal hierarchy for \( M^{D-cKP} \).
\end{thm}

when $m'=1$ and $n_{1}'=1$, this Frobenius manifold structure coincides with that on the orbit space of the Coxeter group of type \( D \) \cite{zuo2007frobenius}, and the associated principal hierarchy is the dispersionless limit of the Drinfeld-Sokolov hierarchy of type \( D \) \cite{liu2011drinfeld}. 

Owing to Theorem 3 in \cite{alcolado2017extended}, we obtain a solution to the open WDVV equations associated with $M^{cKP}$, or equivalently, a flat F-manifold structure on $M^{cKP}\times\mathbb{C}$ which is a rank-1 extension of $M^{cKP}$. The multiplication $\star$ on this flat F-manifold is given by
\begin{align*}
	&(\p_{\alpha},0)\star(\p_{\beta},0)=(\p_{\alpha}\circ\p_{\beta},\p_{\alpha}\p_{\beta}\Omega\cdot\p_{s})\\
	&(\p_{\alpha},0)\star(0,\p_{s})=(0,\p_{\alpha}\lambda(s)\cdot\p_{s})\\
	&(0,\p_{s})\star(0,\p_{s})=(0,\lambda'(s)\cdot\p_{s})
\end{align*}
where $\p_{\alpha}=\frac{\p}{\p t^{\alpha}}$ for $t^{\alpha}\in \textbf{t}^{cKP}$, $\circ$ represents the multiplication on $M^{cKP}$, $s$ is the coordinate on  $\mathbb{C}$, and
$$ \p_{\alpha}\p_{\beta}\Omega=(\frac{\p_{\alpha}\lambda(s)\p_{\beta}\lambda(s)}{\lambda'(s)})_{\infty,\ge 0}+\sum_{j=1}^{m}(\frac{\p_{\alpha}\lambda(s)\p_{\beta}\lambda(s)}{\lambda'(s)})_{\varphi_{j},\le -1}.
$$

\begin{cor}\label{cKPopenpri}
	The principal hierarchy for the flat F-manifold $M^{cKP}\times \mathbb{C}$ is given by
	$$
	\frac{\p\lambda(z)}{\p \tilde{T}^{s,p}}=0,\quad \frac{\p s(x)}{\p \tilde{T}^{s,p}}=d_{x}(\frac{\lambda(s)^{p+1}}{(p+1)!}),
	$$
	and
	\begin{align*}
		&\frac{\p\lambda(z)}{\p\tilde{T}^{t_{0,n_{0}-j},p}}=\frac{\p\lambda(z)}{\p T^{t_{0,j},p}},\quad 	\frac{\p s(x)}{\p\tilde{T}^{t_{0,n_{0}-j},p}}=d_{x}(d\theta_{t_{0,j},p+1}(s))_{+},\\
		&\frac{\p\lambda(z)}{\p\tilde{T}^{t_{i,n_{i}-j},p}}=\frac{\p\lambda(z)}{\p T^{t_{i,j},p}},\quad 	\frac{\p s(x)}{\p\tilde{T}^{t_{0,n_{i}-j},p}}=-d_{x}(d\theta_{t_{i,j},p+1}(s))_{-},
	\end{align*}
	where
	$$
	d_{x}\lambda(s)=\p_{x}\lambda(s)+\lambda'(s)\frac{\p s(x)}{\p x},
	$$
$d\theta_{t_{i,j},p}(z)$ is a certain function on $\lambda(z)$, and the operators $(\ )_{+}, (\ )_{-}$ are defined in Subsection \ref{seccot}.
\end{cor}

When $m=0$, this hierarchy coincides with the dispersionless limit of the open Gelfand-Dickey hierarchy, which is conjectured to govern the generating function of the open r-spin intersection numbers \cite{buryak2024open}.

Let us proceed to consider the Frobenius manifold with trigonometric superpotential. Let $M^{Toda}$ be the space of functions of the form:
\[
\lambda(\varphi) = \frac{1}{n_0}e^{n_0\varphi} + a_{0,n_0-1}e^{(n_0-1)\varphi} + \cdots + a_{0,0} + \sum_{i=1}^{m}\sum_{j=1}^{n_i}a_{i,j}(e^{\varphi}-a_{i,0})^{-j},
\]
with \( a_{1,0} = 0 \), equipped with the Hurwitz Frobenius manifold structure constructed by Dubrovin \cite{dub1998}. In the case of $m=1$, this  structure coincides with that on the orbit space of the extended affine Weyl group  of type \( A \) \cite{dubrovin1998extended}.

The flat coordinate system on \( M^{Toda} \), denoted as
\[
\mathbf{t}^{Toda} = \{t_{0,j}\}_{j=1}^{n_0-1} \cup \{t_{1,j}\}_{j=0}^{n_1} \cup \cdots \cup \{t_{m,j}\}_{j=0}^{n_m},
\]
is given by 
\[
\varphi = \left\{
\begin{aligned}
	&t_{i,0} + t_{i,1}w_{i}(\varphi)^{-1} + \cdots, \quad &e^{\varphi} &\to a_{i,0},\ i = 2, \cdots, m, \\
	&-\log(w_{1}(\varphi)) + t_{1,0} + t_{1,1}w_{1}(\varphi)^{-1} + \cdots, \quad &e^{\varphi} &\to 0, \\
	&\log(w_{0}(\varphi)) - t_{0,1}w_{0}(\varphi)^{-1} - t_{0,2}w_{0}(\varphi)^{-2} - \cdots, \quad &e^{\varphi} &\to \infty,
\end{aligned}
\right.
\]
where
\[
w_{i}(\varphi) = \left\{
\begin{aligned}
	&(n_{i}\lambda(\varphi))^{\frac{1}{n_{i}}} = w_{i,1}(e^{\varphi} - a_{i,0})^{-1} + \cdots, \quad &e^{\varphi} &\to a_{i,0},\ i = 1, \cdots, m, \\
	&(n_{0}\lambda(\varphi))^{\frac{1}{n_{0}}} = e^{\varphi} + w_{0,0} + w_{0,1}e^{-\varphi} + \cdots, \quad &e^{\varphi} &\to \infty,\ i = 0.
\end{aligned}
\right.
\]

\begin{thm}\label{mainToda}
	The Hamiltonian densities of the principal hierarchy for the Frobenius manifold \( M^{Toda} \) are given by
	\begin{align*}
		\theta_{t_{0,j},p}(z) =& -\res_{\infty} c_{0,j;p} w_{0}(z)^{(p+1)n_{0}-j} \frac{dz}{z}, \quad j = 1, \ldots, n_{0} - 1; \\
		\theta_{t_{i,j},p}(z) =& \res_{a_{i,0}} c_{i,j;p} w_{i}(z)^{(p+1)n_{i}-j} \frac{dz}{z}, \quad i = 1, \ldots, m, \ j = 0, \ldots, n_{i} - 1,
	\end{align*}
	and
	\begin{align*}
		\theta_{t_{i,n_{i}},p}(z)=& -\res_{\infty} \frac{\lambda(z)^{p}}{p!} (\log \frac{w_{0}(z)}{z - a_{i,0}} - \frac{c_{p}}{n_{0}}) \frac{dz}{z} + \res_{a_{i,0}} \frac{\lambda(z)^{p}}{p!} (\log(z - a_{i,0}) w_{i}(z) - \frac{c_{p}}{n_{i}}) \frac{dz}{z} \\
		& + \sum_{s \neq i} \res_{a_{s,0}} \frac{\lambda(z)^{p}}{p!} \log(z - a_{i,0}) \frac{dz}{z}, \quad i = 1, \ldots, m,
	\end{align*}
	where $z=e^{\varphi}$. The corresponding Hamiltonian vector fields \( \frac{\partial}{\partial T^{\bullet,p}} = \mathcal{P}(d\theta_{\bullet,p+1}) \) are
	\begin{align*}
		\frac{\partial \lambda(z)}{\partial T^{t_{0,j},p-1}} =& \{(c_{0,j;p-1} w_{0}(z)^{p n_{0}-j})_{\infty, \geq 0}, \lambda(z)\}, \quad j = 1, \ldots, n_{0} - 1; \\
		\frac{\partial \lambda(z)}{\partial T^{t_{i,j},p-1}} =& -\{(c_{i,j;p-1} w_{i}(z)^{p n_{i}-j})_{a_{i,0}, \leq -1}, \lambda(z)\}, \quad i = 1, \ldots, m, \ j = 0, \ldots, n_{i} - 1,
	\end{align*}
	and
		\begin{align*}
		\frac{\partial \lambda(z)}{\partial T^{t_{i,n_{i}},p-1}} =& \{(\frac{\lambda(z)^{p}}{p!} (\log \frac{w_{0}(z)}{z - a_{i,0}} - \frac{c_{p}}{n_{0}}))_{\infty, \geq 0}-(\frac{\lambda(z)^{p}}{p!} (\log(z - a_{i,0}) w_{i}(z) - \frac{c_{p}}{n_{i}}))_{a_{i,0}, \leq -1}, \lambda(z)\}\\
		& - \sum_{s \neq i} \{(\frac{\lambda(z)^{p}}{p!} \log(z - a_{i,0}))_{a_{s,0}, \leq -1}, \lambda(z)\} + \{\frac{\lambda(z)^{p}}{p!} \log(z - a_{i,0}), \lambda(z)\},
	\end{align*}
	where the Poisson bracket is defined as
	\[
	\{f(z,x), g(z,x)\} = z \frac{\partial f(z,x)}{\partial z} \frac{\partial g(z,x)}{\partial x} - z \frac{\partial f(z,x)}{\partial x} \frac{\partial g(z,x)}{\partial z}.
	\]
\end{thm}
When \( m = 1 \) and \( n_0 = n_1 = 1 \), this principal hierarchy coincides with the dispersionless limit of the extended Toda hierarchy \cite{carlet2004extended}, which governs the generating function of the Gromov-Witten invariants of \( \mathbb{CP}^{1} \) \cite{zhang2002cp1}.

\begin{thm}\label{mainTodadul}
	Let $\hat{M}^{Toda}$ be the almost duality of the Frobenius manifold $M^{Toda}$, then the Hamiltonian densities of the principal hierarchy for $\hat{M}^{Toda}$ are given by
	\begin{equation}\label{todaalmham}
		F_{\tilde{\gamma},p} = \frac{1}{2\pi\mathrm{i}} \int_{\tilde{\gamma}} \frac{(\log(\lambda(z)))^{p+1}}{(p+1)!} \frac{dz}{z},
	\end{equation}
	where $\tilde{\gamma}$ is a simple closed curve in the complex plane $\mathbb{C}$, such that the winding number of $\lambda(z)$ along $\tilde{\gamma}$ is zero.
	
	In particular, let $q_{1}, \ldots, q_{k}$ and $p_{1}, \ldots, p_{k}$ be the zeros and poles of $\lambda(z)$, respectively, within the region enclosed by $\tilde{\gamma}$. Additionally, let $q_{k+1}, \ldots, q_{n}, p_{k+1}, \ldots, p_{r}$ be the zeros and poles outside this region. Then we have
	\begin{equation*}
		F_{\tilde{\gamma},0} = \left\{
		\begin{aligned}
			&\sum_{s=1}^{k} \log(p_{s}) - \sum_{s=1}^{k} \log(q_{s}), \quad &\text{if $0 \in \mathbb{C}$ is outside the region enclosed by $\tilde{\gamma}$,} \\
			&\sum_{s=k+1}^{n} \log(p_{s}) - \sum_{s=k+1}^{r} \log(q_{s}), \quad &\text{if $0 \in \mathbb{C}$ is inside the region enclosed by $\tilde{\gamma}$.}
		\end{aligned}
		\right.
	\end{equation*}
\end{thm}

Let's consider the submanifold of $M^{Toda}$ with even superpotential.  Assume there exist positive integers \( m', n_0', \ldots, n_{m'}' \) such that
\[
n_0 = n_1 = n_0', \quad n_2 = 2n_1', \quad n_3 = 2n_2', \quad n_{2j-2} = n_{2j-1} = n_j', \quad m = 2m' - 1,
\]
where \( j = 3, 4, \ldots, m' \). Denote $p=c^{-\frac{1}{n_{0}}}(\frac{1}{n_{0}})^{-\frac{1}{n_{0}}}z$, where $c^{2}n_{0}=a_{1,n_{1}}$, then the elements in $M^{Toda}$ can be expressed as:
$$
\lambda(p)=\sum_{j=0}^{n_{0}}b_{0,j}p^{j}+\sum_{j=1}^{n_{0}}b_{1,j}p^{-j}+\sum_{i=1}^{m}\sum_{j=1}^{n_{i}}b_{i,j}(p-b_{i,0})^{-j},
$$
where $b_{1,n_{0}}=b_{0,n_{0}}$. Let \( M^{C-Toda} \) be the space of elements in \( M^{Toda} \) of the following form:
\begin{equation}\label{Ctodasup}
	\lambda(p) = \sum_{j=0}^{n_0'} \tilde{b}_{0,j} (p + \frac{1}{p})^j + \sum_{j=1}^{n_1'} \tilde{b}_{1,j} (p + \frac{1}{p} - 2)^{-j} + \sum_{j=1}^{n_2'} \tilde{b}_{2,j} (p + \frac{1}{p} + 2)^{-j} + \sum_{i=3}^{m'} \sum_{j=1}^{n_i'} \tilde{b}_{i,j} (p + \frac{1}{p} - \tilde{b}_{i,0})^{-j},
\end{equation}
which are characterized by $\lambda(p)=\lambda(\frac{1}{p})$.

\begin{thm}\label{mainTodasym}
	$M^{C-Toda}$ is a natural Frobenius submanifold of $M^{Toda}$, characterized by the conditions imposed on the flat coordinates of $M^{Toda}$ as follows:
	\begin{align*}
		&t_{0,j}=t_{1,j},\quad j=1,\cdots,n_{0}-1;\\
		&t_{2,2}=t_{2,4}=\cdots= t_{2,n_{2}}=0,\quad t_{2,0}-\frac{1}{2}t_{1,0}=-\log 1;\\
		&t_{3,2}=t_{3,4}=\cdots= t_{3,n_{3}}=0,\quad t_{3,0}-\frac{1}{2}t_{1,0}=-\log (-1);\\
		&t_{2i-2,j}=t_{2i-1,j},\quad 3\le i\le m',\ 0\le j\le n_{i}',
	\end{align*}
	where the constants \(\log 1\) and \(\log(-1)\) depend on the choice of branch of \(\log p\). 	Moreover, the subhierarchy \( \frac{\partial}{\partial T^{t,p}} \) of the principal hierarchy for \( M^{Toda} \), where
	\[
	t \in \textbf{t}^{C-Toda} = \{t_{0,j}\}_{j=1}^{n_{0}'-1} \cup \{t_{2,2j-1}\}_{j=1}^{n_{1}'} \cup \{t_{3,2j-1}\}_{j=1}^{n_{2}'} \cup \{t_{4,j}\}_{j=0}^{n_{3}'} \cup \cdots \cup \{t_{2m'-2,j}\}_{j=0}^{n_{m'}'},
	\]
	can be directly restricted to the submanifold \( M^{C-Toda} \),  forming the principal hierarchy for \( M^{C-Toda} \).
\end{thm}
when $m'=2$ and $n_{1}'=n_{2}'=2$, this Frobenius manifold structure coincides with that on the orbit space of the extended affine Weyl group of type \( D \) \cite{dubrovin1998extended,dubrovin2019extended}. As far as we know, the explicit form of the Dubrovin-Zhang hierarchy associated with this Frobenius manifold remains unknown. However, Minanov and Cheng proposed bilinear-type equations governing the descendant potential of this Frobenius manifold \cite{cheng2021hirota}, and reformulated these equations in the form of Lax equations \cite{cheng2021extended}. 

There exists a flat F-manifold structure on $M^{Toda}\times\mathbb{C}$ which is a rank-1 extension of $M^{Toda}$. The multiplication on this flat F-manifold is given by
\begin{align*}
	&(\p_{\alpha},0)\star(\p_{\beta},0)=(\p_{\alpha}\circ\p_{\beta},\p_{\alpha}\p_{\beta}\Omega\cdot\p_{s})\\
	&(\p_{\alpha},0)\star(0,\p_{s})=(0,\p_{\alpha}\lambda(z)\cdot\p_{s})\\
	&(0,\p_{s})\star(0,\p_{s})=(0,z\lambda'(z)\cdot\p_{s})
\end{align*}
where $s$ is the coordinate on  $\mathbb{C}$, $z=e^{s}$,  $\p_{\alpha}=\frac{\p}{\p t^{\alpha}}$ for $t^{\alpha}\in \textbf{t}^{Toda}$, and
$$ \p_{\alpha}\p_{\beta}\Omega=(\frac{\p_{\alpha}\lambda(z)\p_{\beta}\lambda(z)}{z\lambda'(z)})_{\infty,\ge 0}+\sum_{j=1}^{m}(\frac{\p_{\alpha}\lambda(z)\p_{\beta}\lambda(z)}{z\lambda'(z)})_{\varphi_{j},\le -1}.
$$

\begin{cor}\label{Todaopenpri}
	The principal hierarchy for the flat F-manifold $M^{Toda}\times \mathbb{C}$ is given by
	$$
	\frac{\p\lambda(z)}{\p \tilde{T}^{s,p}}=0,\quad \frac{\p s(x)}{\p \tilde{T}^{s,p}}=d_{x}(\frac{\lambda(z)^{p+1}}{(p+1)!}),
	$$
	and
	\begin{align*}
		&\frac{\p\lambda(z)}{\p\tilde{T}^{t_{0,n_{0}-j},p}}=\frac{\p\lambda(z)}{\p T^{t_{0,j},p}},\quad 	\frac{\p s(x)}{\p\tilde{T}^{t_{0,n_{0}-j},p}}=d_{x}(d\theta_{t_{0,j},p+1}(z))_{+},\\
		&\frac{\p\lambda(z)}{\p\tilde{T}^{t_{i,n_{i}-j},p}}=\frac{\p\lambda(z)}{\p T^{t_{i,j},p}},\quad 	\frac{\p s(x)}{\p\tilde{T}^{t_{0,n_{i}-j},p}}=-d_{x}(d\theta_{t_{i,j},p+1}(z))_{-},
	\end{align*}
	where
	$$
	d_{x}\lambda(z)=\p_{x}\lambda(z)+z\lambda'(z)\frac{\p s(x)}{\p x},
	$$
	$d\theta_{t_{i,j},p}(z)$ is certain function on $\lambda(z)$.
\end{cor}

We hope that this result can provide insight into defining an open-type extended Toda hierarchy governing the generating function of the open Gromov-Witten invariants of \( \mathbb{CP}^{1} \) \cite{buryak2022open}.

This paper is organized as follows: In Section 2, we review the definition of Frobenius manifold, almost duality, flat F-manifold, and associated principal hierarchy. In Section 3, we construct the principal hierarchy for the Frobenius manifold with rational superpotential and its almost duality. We then show that this principal hierarchy can be directly restricted to the Frobenius submanifold with even superpotential. Finally, we provide a natural rank-1 extension of this Frobenius manifold and construct the associated principal hierarchy. In Section 4, we apply the same procedure to the Frobenius manifold with trigonometric superpotential.

\section{preliminary on Frobenius manifold}
In this section, we will recall the definition of Frobenius manifold and associated principal hierarchy.
\subsection{Frobenius manifold and principal hierarchy}
A Frobenius manifold of charge \(d\) is an \(n\)-dimensional manifold \(M\), where each tangent space \(T_v M\) is equipped with a Frobenius algebra structure \((A_v = T_v M, \circ, e, \langle \cdot, \cdot \rangle)\) that varies smoothly with \(v \in M\). This structure satisfies the following axioms:
\begin{enumerate}
	\item The bilinear form \(\langle \cdot, \cdot \rangle\) provides a flat metric on \(M\), and the unity vector field \(e\) satisfies \(\nabla e = 0\), where \(\nabla\) is the Levi-Civita connection for the flat metric.
	\item Define a 3-tensor \(c\) by \(c(X, Y, Z) := \langle X \circ Y, Z \rangle\) with \(X, Y, Z \in T_v M\). Then, the 4-tensor \((\nabla_W c)(X, Y, Z)\) is symmetric in \(X, Y, Z, W \in T_v M\).
	\item There exists a vector field \(E\), called  the Euler vector field, which satisfies \(\nabla^2 E = 0\) and 
	\[
	\begin{aligned}
		& [E, X \circ Y] - [E, X] \circ Y - X \circ [E, Y] = X \circ Y, \\
		& E(\langle X, Y \rangle) - \langle [E, X], Y \rangle - \langle X, [E, Y] \rangle = (2 - d) \langle X, Y \rangle
	\end{aligned}
	\]
	for any vector fields $X,Y$ on $M$.
\end{enumerate}

On an \(n\)-dimensional Frobenius manifold \(M\), we select a set of flat coordinates \(t = (t^1, \ldots, t^n)\) such that \(e = \frac{\partial}{\partial t^1}\). In this coordinate system, the components of the metric $\langle\ ,\ \rangle$ are given by:
\[
\eta_{\alpha \beta} = \left\langle \frac{\partial}{\partial t^\alpha}, \frac{\partial}{\partial t^\beta} \right\rangle, \quad \alpha, \beta = 1, \ldots, n,
\]
where 
$\eta_{\alpha\beta}$ defines a constant and non-degenerate $n\times n$ matrix.
The inverse of this matrix is denoted by \(\eta^{\alpha \beta}\). The metric and its inverse are utilized for index lowering and raising, respectively, with the Einstein summation convention applied to repeated Greek indices.

Furthermore, we denote the components of the 3-tensor \(c\) by:
\[
c_{\alpha \beta \gamma} = c\left(\frac{\partial}{\partial t^\alpha}, \frac{\partial}{\partial t^\beta}, \frac{\partial}{\partial t^\gamma}\right), \quad \alpha, \beta, \gamma = 1, \ldots, n,
\]
which allows us to express the multiplication structure of the Frobenius algebra \(T_v M\) in terms of 
\[
\frac{\partial}{\partial t^\alpha} \circ \frac{\partial}{\partial t^\beta} = c_{\alpha \beta}^\gamma \frac{\partial}{\partial t^\gamma},
\]
where the coefficients \(c_{\alpha \beta}^\gamma\) are obtained by contracting the 3-tensor \(c\) with the metric \(\eta^{\alpha\beta}\):
\[
c_{\alpha \beta}^\gamma = \eta^{\gamma \epsilon} c_{\epsilon \alpha \beta},
\]
which satisfy
\[
c_{1 \alpha}^\beta = \delta_\alpha^\beta, \quad c_{\alpha \beta}^\epsilon c_{\epsilon \gamma}^\sigma = c_{\alpha \gamma}^\epsilon c_{\epsilon \beta}^\sigma.
\]

According to the definition of Frobenius manifold, there exists a smooth function \( F(t) \) satisfying the following properties:
\[
\begin{aligned}
	c_{\alpha \beta \gamma} & = \frac{\partial^3 F}{\partial t^\alpha \partial t^\beta \partial t^\gamma}, \\
	\operatorname{Lie}_E F & = (3-d) F + \text{quadratic terms in } t.
\end{aligned}
\]
Hence, \( F(t) \) is a solution to the WDVV equation
\[
\frac{\partial^3 F}{\partial t^\alpha \partial t^\beta \partial t^\gamma} \eta^{\gamma \epsilon} \frac{\partial^3 F}{\partial t^\epsilon \partial t^\sigma \partial t^\mu} = \frac{\partial^3 F}{\partial t^\alpha \partial t^\sigma \partial t^\gamma} \eta^{\gamma \epsilon} \frac{\partial^3 F}{\partial t^\epsilon \partial t^\beta \partial t^\mu}.
\]
The third-order derivatives \( c_{\alpha \beta \gamma} \) of \( F(t) \) are known as the 3-point correlator functions in the context of topological field theory.

For a Frobenius manifold \(M\), its cotangent space \(T_v^* M\) is endowed with a Frobenius algebra structure as well. This structure encompasses an invariant bilinear form and a product, which are defined by:
\[
\left\langle d t^\alpha, d t^\beta\right\rangle = \eta^{\alpha \beta}, \quad d t^\alpha \circ d t^\beta = \eta^{\alpha \epsilon} c_{\epsilon \gamma}^\beta.
\]
Let us define
\[
g^{\alpha \beta} = i_E\left(d t^\alpha \circ d t^\beta\right),
\]
then \( \left(d t^\alpha, d t^\beta\right) := g^{\alpha \beta} \) establishes a symmetric bilinear form known as the intersection form on \(T_v^* M\).
The intersection form \(g^{\alpha \beta}\) and the invariant bilinear form \(\eta^{\alpha \beta}\)  together form a pencil of flat metrics 
\[
g^{\alpha \beta} + \epsilon \eta^{\alpha \beta}
\]
parameterized by \( \epsilon \). As a result, they give rise to a bi-hamiltonian structure of hydrodynamic type on the loop space \( \{S^1 \rightarrow M\} \), expressed as:
\[
\{\ ,\ \}_2 + \epsilon \{\ ,\ \}_1.
\]

The deformed flat connection on \(M\), originally introduced by Dubrovin \cite{dub1998}, is defined as:
\[
\widetilde{\nabla}_X Y = \nabla_X Y + z X \circ Y, \quad X, Y \in \operatorname{Vect}(M).
\]
This connection can be consistently extended to a flat affine connection on \(M \times \mathbb{C}^*\) such that
\[
\begin{aligned}
	&\tilde{\nabla}_X \frac{d}{d z} = 0, \\
	&\tilde{\nabla}_{\frac{d}{d z}} \frac{d}{d z} = 0, \\
	&\tilde{\nabla}_{\frac{d}{d z}} X = \partial_z X + E \circ X - \frac{1}{z} \mathcal{V}(X),
\end{aligned}
\]
where \(X\) is a vector field on \(M \times \mathbb{C}^*\) that vanishes in the \(\mathbb{C}^*\) component, and \(\mathcal{V}(X)\) is defined as
\[
\mathcal{V}(X) := \frac{2-d}{2} X - \nabla_X E.
\]

There exists a system of deformed flat coordinates \(\tilde{v}_1(t, z), \ldots, \tilde{v}_n(t, z)\) that can be expressed in terms of
\[
\left(\tilde{v}_1(t, z), \ldots, \tilde{v}_n(t, z)\right) = \left(\theta_1(t, z), \ldots, \theta_n(t, z)\right) z^\mu z^R.
\]
These coordinates are chosen such that the 1-forms
\[
\xi_\alpha = \frac{\partial \tilde{v}_\alpha}{\partial t^\beta} d t^\beta, \quad \alpha = 1, \ldots, n, \quad \text{and} \quad \xi_{n+1} = d z,
\]
constitute a basis of solutions to the system \( \widetilde{\nabla} \xi = 0 \). Here, \( \mu = \operatorname{diag}(\mu_1, \ldots, \mu_n) \) is a diagonal matrix characterized by
\[
\mathcal{V}\left(\frac{\partial}{\partial t^\alpha}\right)=\mu_\alpha \frac{\partial}{\partial t^\alpha}, \quad \alpha = 1, \ldots, n,
\]
which is called the spectrum of $M$,
and \( R = R_1 + \ldots + R_m \) is a constant nilpotent matrix satisfying
\[
\begin{aligned}
	& (R_s)_\beta^\alpha = 0 \text{ if } \mu_\alpha - \mu_\beta \neq s, \\
	& (R_s)_\alpha^\gamma \eta_{\gamma \beta} = (-1)^{s+1} (R_s)_\beta^\gamma \eta_{\gamma \alpha}.
\end{aligned}
\]
The functions \( \theta_\alpha(t, z) \), being analytic near \( z = 0 \), can be represented by a power series expansion:
\[
\theta_\alpha(t, z) = \sum_{p \geq 0} \theta_{\alpha, p}(t) z^p, \quad \alpha = 1, \ldots, n.
\]
The coefficients of this expansion satisfy
\begin{equation}\label{princon1}
	\frac{\partial^2 \theta_{\alpha, p+1}(t)}{\partial t^\beta \partial t^\gamma} = c_{\beta \gamma}^\epsilon(t) \frac{\partial \theta_{\alpha, p}(t)}{\partial t^\epsilon},
\end{equation}
and
\begin{equation}\label{princon2}
	\operatorname{Lie}_E\left(\frac{\partial \theta_{\alpha, p}(t)}{\partial t^\beta}\right) = \left(p + \mu_\alpha + \mu_\beta\right) \frac{\partial \theta_{\alpha, p}(t)}{\partial t^\beta} + \frac{\partial}{\partial t^\beta} \sum_{s=1}^p \theta_{\epsilon, p-s}(t) \left(R_{s}\right)_\alpha^\epsilon.
\end{equation}
Moreover, the normalization condition\footnote{In the referenced work \cite{dub1998}, an additional condition \( \langle \nabla \theta_\alpha(t, z), \nabla \theta_\beta(t, -z) \rangle = \eta_{\alpha \beta} \) was considered. However, as it does not significantly alter the properties of the principal hierarchy, we omit it here for computational simplicity} is imposed:
\begin{equation}\label{princon3}
	\theta_{\alpha, 0}(t) = \eta_{\alpha \beta} t^\beta.
\end{equation}

Given a system of solutions \(\left\{\theta_{\alpha, p}\right\}\) to the equations \eqref{princon1}-\eqref{princon3}, the principal hierarchy associated with \(M\) is defined as the following Hamiltonian system on the loop space \(\left\{S^1 \rightarrow M\right\}\) :
\[
\frac{\partial t^\gamma}{\partial T^{\alpha, p}} = \left\{t^\gamma(x), \int \theta_{\alpha, p+1} (t) \, d x\right\}_1 := \eta^{\gamma \beta} \frac{\partial}{\partial x}\left(\frac{\partial \theta_{\alpha, p+1} (t)}{\partial t^\beta}\right), \quad \alpha, \beta = 1,2, \ldots, n, \ p \geq 0.
\]
These commuting flows are tau-symmetric, which means that
\[
\frac{\partial \theta_{\alpha, p}(t)}{\partial T^{\beta, q}} = \frac{\partial \theta_{\beta, q}(t)}{\partial T^{\alpha, p}}, \quad \alpha, \beta = 1,2, \ldots, n, \ p, q \geq 0.
\]
Furthermore, these flows can be expressed in a bi-hamiltonian recursion form as
\[
\mathcal{R} \frac{\partial}{\partial T^{\alpha, p-1}} = \frac{\partial}{\partial T^{\alpha, p}}\left(p + \mu_\alpha + \frac{1}{2}\right) + \sum_{s=1}^p \frac{\partial}{\partial T^{\epsilon, p-s}}\left(R_{s}\right)_\alpha^\epsilon,
\]
where \(\mathcal{R} = \{\ ,\ \}_2 \cdot \{\ ,\ \}_1^{-1}\).
\subsection{almost duality}
The discriminant of a Frobenius manifold is defined as
\[
\Sigma = \{p \in M | \det(g^{\alpha\beta}(p)) = 0\}.
\]
Let \( \hat{M} = M \setminus \Sigma \), and for any \( p \in \hat{M} \), define the multiplication on \( T_{p}\hat{M} \) as
\[
u \star v = E^{-1} \circ u \circ v, \quad u, v \in T_{p}\hat{M}.
\]
Then the data set \( (\hat{M}, g^{\alpha\beta}, \star) \) satisfies the axiom (2) in the definition of Frobenius manifold. This data set is referred to as the almost duality \cite{dubrovin2004almost} of the Frobenius manifold \( M \).

Let \( \hat{t} = (\hat{t}^{1}, \ldots, \hat{t}^{n}) \) be the flat coordinates of \( g^{\alpha\beta} \), and let \( \hat{c}_{\alpha\beta}^{\gamma}(t) \) be the components of the multiplication structure \( \star \) in these flat coordinates. Consider the formal power series
\[
\hat{\theta}_\alpha(\hat{t}, z) = \sum_{p \geq 0} \hat{\theta}_{\alpha, p}(\hat{t}) z^p, \quad \alpha = 1, \ldots, n,
\]
such that the following system of equations holds:
\[
\frac{\partial^2 \hat{\theta}_{\alpha, p+1}(\hat{t})}{\partial \hat{t}^\beta \partial \hat{t}^\gamma} = \hat{c}_{\beta \gamma}^\epsilon(\hat{t}) \frac{\partial \hat{\theta}_{\alpha, p}(\hat{t})}{\partial \hat{t}^\epsilon}, \quad \hat{\theta}_{\alpha,0} = g_{\alpha\beta}\hat{t}^{\beta},
\]
then \( d\hat{\theta}(\hat{t},z) \) provides a set of fundamental solutions near \( z = 0 \) for the deformed flat connection: 
\[
\hat{\nabla}^{deform}_{u}v = \hat{\nabla}_{u}v + z \cdot E^{-1} \circ u \circ v,
\]
where \( \hat{\nabla} \) is the Levi-Civita connection of \( g^{\alpha\beta} \). The Hamiltonian system
\[
\frac{\partial \hat{t}^\gamma}{\partial \hat{T}^{\alpha, p}} = \left\{\hat{t}^\gamma(x), \int \hat{\theta}_{\alpha, p+1} (\hat{t}) d x\right\}_2 := g^{\gamma \beta} \frac{\partial}{\partial x}\left(\frac{\partial \hat{\theta}_{\alpha, p+1} (\hat{t})}{\partial \hat{t}^\beta}\right), \quad \alpha, \beta = 1,2, \ldots, n, \ p \geq 0,
\]
defines an integrable hierarchy on the loop space of \( \hat{M} \).

\subsection{flat F-manifold}
A flat F-manifold $(M,\nabla,\circ,e)$ consists of an analytic manifold $M$, a flat torsionless connection $\nabla$ on $TM$, and a commutative associated algebra structure on each tangent space $T_{p}M$ with a unit vector field $e$, satisfying the following conditions:
\begin{enumerate}
	\item[(1)] $\nabla e=0$;
	\item[(2)] There exists a vector field $\Psi$ on $M$, called the vector potential, such that
	$$
	X\circ Y=[X,[Y,\Psi]]
	$$
	for any flat vector field $X$ and $Y$ on $M$.
\end{enumerate}
For more details, see \cite{manin2005f,arsie2018flat,alcolado2017extended}.

Let $\textbf{t} = (t^{1}, \dots, t^{n})$ denote the flat coordinates associated with the connection $\nabla$. The principal hierarchy for the flat F-manifold $M$ is defined as:
$$
\frac{\p}{\p T^{i,p}}=\Theta_{i,p}\circ \p_{x},\quad i=1,\cdots,n,
$$
where the vector fields $\Theta_{i,p}$ on $M$ satisfy
$$
\Theta_{i,0}=\frac{\p}{\p t^{i}},\quad i=1,\cdots,n,
$$
and
$$
\nabla_{X}\Theta_{i,p+1}=\Theta_{i,p}\circ X
$$
for any vector field $X$ on $M$.

Let $M$ be a Frobenius manifold with flat coordinates $\textbf{t}=(t^{1},\cdots,t^{n})$ and potential $F(\textbf{t})$. There exists a flat F-manifold structure on $M\times \mathbb{C}$ with vector potential:
$$
\Psi=\frac{\p F(\textbf{t})}{\p t^{\alpha}}\eta^{\alpha\beta}\frac{\p}{\p t^{\beta}}+\Omega(\textbf{t},s)\frac{\p}{\p s}
$$
if and only if $\Omega(\textbf{t},s)$ satisfies the  open WDVV equations:
$$
c_{\alpha\beta}^{\delta}\frac{\p\Omega(\textbf{t},s)}{\p t^{\delta}\p t^{\gamma}}+\frac{\p\Omega(\textbf{t},s)}{\p t^{\alpha}\p t^{\beta}}\frac{\p\Omega(\textbf{t},s)}{\p t^{\gamma}\p s}=c_{\beta\gamma}^{\delta}\frac{\p\Omega(\textbf{t},s)}{\p t^{\delta}\p t^{\alpha}}+\frac{\p\Omega(\textbf{t},s)}{\p t^{\beta} \p t^{\gamma}}\frac{\p\Omega(\textbf{t},s)}{\p t^{\alpha}\p s}
$$
and
$$
c_{\alpha\beta}^{\delta}\frac{\p\Omega(\textbf{t},s)}{\p  t^{\delta} \p s}+\frac{\p\Omega(\textbf{t},s)}{\p t^{\alpha} \p t^{\beta}}\frac{\p\Omega(\textbf{t},s)}{\p s\p s}=\frac{\p\Omega(\textbf{t},s)}{\p t^{\alpha} \p s}\frac{\p\Omega(\textbf{t},s)}{\p t^{\beta}\p s},
$$
where $s$ denotes the coordinate on $\mathbb{C}$.

The following lemma by A. Alcolado is useful for constructing such a function $\Omega(\textbf{t}, s)$.
\begin{lem}[\cite{alcolado2017extended}]\label{consopen}
	Let $\omega=\omega(\textbf{t},s)$ be a smooth function on $M\times \mathbb{C}$ satisfying $\p_{s} \omega\ne 0$ and
	$$
	\p_{\alpha}\p_{\beta}\omega=\p_{s}(\frac{\p_{\alpha}\omega \p_{\beta}\omega-c_{\alpha\beta}^{\delta}\p_{\delta}\omega}{\p_{s}\omega}),\quad \alpha,\beta=1,\cdots,n,
	$$
	where $\p_{\alpha}=\frac{\p}{\p t^{\alpha}}$. Then, the function $\Omega(\textbf{t},s)$ defined by  $\omega=\p_{s}\Omega$ provides a solution to the open WDVV equations associated with $M$.
\end{lem}

\section{Frobenius manifold with rational  superpotential}
\subsection{Definition of $M^{cKP}$}

Given positive integers $m$ and \( n_0, \ldots, n_m \), let \(  M^{cKP} \) be the space of rational functions
\[
\lambda(z) = \frac{1}{n_0}z^{n_0} + a_{0,n_0-2}z^{n_0-2} + \cdots + a_{0,0} + \sum_{i=1}^{m}\sum_{j=1}^{n_i}a_{i,j}(z-a_{i,0})^{-j}
\]
where \( a_{i,n_i} \neq 0, \ i = 1, \ldots, m \). The parameters \( \{a_{0,i}\}_{i=0}^{n_0-2} \cup \{a_{1,i}\}_{i=0}^{n_1} \cup \cdots \cup \{a_{m,i}\}_{i=0}^{n_m} \) form a coordinate system on \( M^{cKP} \). 

For any \( \p', \p'', \p''' \in T_{\lambda(z)}M^{cKP} \), define the metric 
\[
\langle \p', \p'' \rangle_{\eta} := \eta(\p', \p'') = \sum_{|\lambda|<\infty} \res_{d\lambda=0} \frac{\p'(\lambda(z)dz) \p''(\lambda(z)dz)}{d\lambda(z)}
\]
and the \( (0,3) \)-type tensor
\[
c(\p', \p'', \p''') := \sum_{|\lambda|<\infty} \res_{d\lambda=0} \frac{\p'(\lambda(z)dz) \p''(\lambda(z)dz) \p'''(\lambda(z)dz)}{d\lambda(z)dz},
\]
then the equality
\[
c(\p', \p'', \p''') = \eta(\p' \circ \p'', \p''')
\]
defines the multiplication $\circ$ on \( T_{\lambda(z)}M^{cKP} \).  Introduce vector fields \( e \) and \( E \) on \( M^{cKP}\) such that
\[
Lie_{e}\lambda(z) = 1, \quad Lie_{E}\lambda(z) = \lambda(z) - \frac{z}{n_{0}}\lambda'(z),
\]
then the data set \( (M, \eta, \circ, e, E) \) constitutes a semisimple Frobenius manifold with \( d = 1 - \frac{2}{n_{0}} \). 

The flat coordinates of the metric \( \eta \), denoted as
\[
\mathbf{t} = \{t_{0,j}\}_{j=1}^{n_{0}-1} \cup \{t_{1,j}\}_{j=0}^{n_{1}} \cup \cdots \cup \{t_{m,j}\}_{j=0}^{n_{m}},
\]
are given by the coefficients of the following series:
\[
z = \left\{
\begin{aligned}
	&t_{i,0} + t_{i,1}w_{i}^{-1} + \cdots, \quad &z &\to a_{i,0},\ i = 1, \cdots, m, \\
	&w_{0} - t_{0,1}w_{0}^{-1} - t_{0,2}w_{0}^{-2} + \cdots, \quad &z &\to \infty,\ \ i = 0,
\end{aligned}
\right.
\]
where
\[
w_{i} = \left\{
\begin{aligned}
	&(n_{i}\lambda)^{\frac{1}{n_{i}}} = w_{i,1}(z - a_{i,0})^{-1} + \cdots, \quad &z &\to a_{i,0},\ i = 1, \cdots, m, \\
	&(n_{0}\lambda)^{\frac{1}{n_{0}}} = z + w_{0,1}z^{-1} + \cdots, \quad &z &\to \infty,\ i = 0.
\end{aligned}
\right.
\]
Furthermore, we have
\[
\frac{\partial \lambda(z)}{\partial t_{i,j}} = \left\{
\begin{aligned}
	&-(w_{i}(z)^{n_{i}-j-1}w_{i}'(z))_{a_{i,0},\le -1},\quad i = 1, \cdots, m,\ j = 0, \cdots, n_{i}, \\
	&(w_{i}(z)^{n_{0}-j-1}w_{0}'(z))_{\infty,\ge 0},\quad i = 0,\ j = 1, \cdots, n_{0}-1.
\end{aligned}
\right.
\]
In the flat coordinate system, the vector fields \(e\) and \( E \) can be expressed as
\[
e=\frac{\p}{\p t_{0,n_{0}-1}},\quad E = \sum_{j=1}^{n_{0}-1} \left(\frac{1+j}{n_{0}}\right)t_{0,j}\frac{\partial}{\partial t_{0,j}} + \sum_{i=1}^{m} \sum_{j=0}^{n_{i}} \left(\frac{1}{n_{0}} + \frac{j}{n_{i}}\right)t_{i,j}\frac{\partial}{\partial t_{i,j}},
\]
thus the spectrum of \( M^{cKP} \) is $
\mu_{t_{i,j}} = \frac{1}{2} - \frac{j}{n_{i}}$.

\begin{lem}\label{kdvconexp}
	Let \( \nabla \) be the Levi-Civita connection associated with the metric \( \eta \). Then, for any vector fields \( \p_1 \) and \( \p_2 \) on \( M^{cKP} \), it holds that
	\begin{equation}\label{kdvconn}
		(\nabla_{\p_1}\p_2) \cdot \lambda(z) = \p_1 \p_2 \lambda(z) - \left(\frac{\p_1 \lambda(z) \p_2 \lambda(z)}{\lambda'(z)}\right)'_{\infty, \ge 0} - \sum_{s=1}^{m} \left(\frac{\p_1 \lambda(z) \p_2 \lambda(z)}{\lambda'(z)}\right)'_{a_{s,0}, \le -1}.
	\end{equation}
\end{lem}

\begin{proof}
	It suffices to prove that the connection \( \nabla \) defined by \eqref{kdvconn} satisfies
	\begin{equation}\label{kdvtor}
		(\nabla_{\p_1}\p_2) \cdot \lambda(z) - (\nabla_{\p_2}\p_1) \cdot \lambda(z) = \p_1 \p_2 \lambda(z) - \p_2 \p_1 \lambda(z)
	\end{equation}
	and
	\begin{equation}\label{conncamp}
		\nabla_{\p_1}\left\langle\p_2, \p_3\right\rangle_\eta = \left\langle\nabla_{\p_1} \p_2, \p_3\right\rangle_\eta + \left\langle\p_2, \nabla_{\p_1} \p_3\right\rangle_\eta.
	\end{equation}
	The equality \eqref{kdvtor} can be easily verified. For the left-hand side of the equality \eqref{conncamp}, we have
	\begin{equation}
		\begin{aligned}
			& \nabla_{\p_1}\left\langle\p_2, \p_3\right\rangle_\eta \\
			=& -\res_{\infty}\left( \frac{\p_1 \p_2 \lambda(z) \cdot \p_3 \lambda(z)}{\lambda'(z)} + \frac{\p_2 \lambda(z) \cdot \p_1 \p_3 \lambda(z)}{\lambda'(z)} - \frac{\p_2 \lambda(z) \cdot \p_3 \lambda(z) \cdot \p_1 \lambda'(z)}{(\lambda'(z))^2} \right)dz \\
			& -\sum_{s=1}^{m}\res_{a_{s,0}}\left( \frac{\p_1 \p_2 \lambda(z) \cdot \p_3 \lambda(z)}{\lambda'(z)} + \frac{\p_2 \lambda(z) \cdot \p_1 \p_3 \lambda(z)}{\lambda'(z)} - \frac{\p_2 \lambda(z) \cdot \p_3 \lambda(z) \cdot \p_1 \lambda'(z)}{(\lambda'(z))^2} \right)dz.
		\end{aligned}
	\end{equation}
	Because \( \p_1, \p_2 \) commute with \( \p_z \), applying integration by parts, we obtain
	\begin{equation}
		\begin{aligned}
			& \res_{\infty} \frac{\p_2 \lambda(z) \cdot \p_3 \lambda(z) \cdot \p_1 \lambda'(z)}{(\lambda'(z))^2} dz \\
			= &  \res_{\infty} \frac{\p_2 \lambda(z)}{\lambda'(z)}\left(\frac{\p_3 \lambda(z) \cdot \p_1 \lambda(z)}{\lambda'(z)}\right)' dz + \res_{\infty} \frac{\p_3 \lambda(z)}{\lambda'(z)}\left(\frac{\p_2 \lambda(z) \cdot \p_1 \lambda(z)}{\lambda'(z)}\right)' dz.
		\end{aligned}
	\end{equation}
	Thus, 
	\begin{equation}
		\begin{aligned}
			& \nabla_{\p_1}\left\langle\p_2, \p_3\right\rangle_\eta \\
			=& -\res_{\infty} \left(\frac{\p_{1} \p_{2} \lambda(z) \cdot \p_{3} \lambda(z)}{\lambda'(z)} + \frac{\p_{2} \lambda(z) \cdot \p_{1} \p_{3} \lambda(z)}{\lambda'(z)}\right)dz\\
			& +\res_{\infty}\left(\frac{\p_{2 }\lambda(z)}{\lambda'(z)}\left(\frac{\p_{3} \lambda(z) \cdot \p_{1} \lambda(z)}{\lambda'(z)}\right)' + \frac{\p_3 \lambda(z)}{\lambda'(z)}\left(\frac{\p_2 \lambda(z) \cdot \p_1 \lambda(z)}{\lambda'(z)}\right)'\right) dz \\
			& -\sum_{s=1}^{m}\res_{a_{s,0}}\left( \frac{\p_1 \p_2 \lambda(z) \cdot \p_3 \lambda(z)}{\lambda'(z)} + \frac{\p_2 \lambda(z) \cdot \p_1 \p_3 \lambda(z)}{\lambda'(z)}\right)dz\\
			& +\sum_{s=1}^{m}\res_{a_{s,0}}\left(\frac{\p_2 \lambda(z)}{\lambda'(z)}\left(\frac{\p_3 \lambda(z) \cdot \p_1 \lambda(z)}{\lambda'(z)}\right)' + \frac{\p_3 \lambda(z)}{\lambda'(z)}\left(\frac{\p_2 \lambda(z) \cdot \p_1 \lambda(z)}{\lambda'(z)}\right)'\right) dz \\
			=& \left\langle\nabla_{\p_1} \p_2, \p_3\right\rangle_\eta + \left\langle\p_2, \nabla_{\p_1} \p_3\right\rangle_\eta.
		\end{aligned}
	\end{equation}
		The lemma is proved.
\end{proof}

\subsection{cotangent space}\label{seccot}
For any point \( \lambda(z) \) in \(M^{cKP} \), a tangent vector \( \p \in T_{\lambda(z)}M^{cKP} \) can be expressed as a rational function:
\[
\p(\lambda(z)) = b_{0,n_{0}-2}z^{n_{0}-2} + \cdots + b_{0,0} + \sum_{i=1}^{m} \sum_{j=1}^{n_{i}+1} b_{i,j}(z-a_{i,0})^{-j}.
\]

To describe a cotangent vector at \( \lambda(z) \), we consider a collection of disjoint disks \( D_1, \ldots, D_m \) on the complex plane \( \mathbb{C} \), such that \( a_{i,0} \in D_i\) for \(i = 1, \ldots, m \). Denote \( \gamma_i = \partial D_i \), \( \textbf{D} = \cup_{s=1}^{m} D_s \), and \( \textbf{D}^c = \mathbb{P}^1 \setminus \textbf{D} \). Let \( \mathcal{H} \) be the space of germs of holomorphic functions on the curves \( \cup_{s=1}^{m} \gamma_s \). For any \( f(z) \in \mathcal{H} \), define
\[
f(z)_+ := \frac{1}{2\pi \mathrm{i}} \sum_{s=1}^{m} \int_{\gamma_s} \frac{f(p)}{p-z} \, dp, \quad z \in \mathring{\mathbf{D}}, \quad f(z)_- := -\frac{1}{2\pi \mathrm{i}} \sum_{s=1}^{m} \int_{\gamma_s} \frac{f(p)}{p-z} \, dp, \quad z \in \mathbf{D}^c,
\]
where \( f(z)_+ \) and \( f(z)_- \) are holomorphic on \( \textbf{D} \) and \( \textbf{D}^c \), respectively, and can be analytically continued to some neighborhood of \( \cup_{s=1}^{m} \gamma_s \). Thus, \( f(z)_+, f(z)_- \in \mathcal{H} \), and \( f(z) = f(z)_+ + f(z)_- \). Conversely, if there exists a decomposition \( f(z) = f_1(z) + f_2(z) \), where \( f_1(z) \) and \( f_2(z) \) are holomorphic on \( \textbf{D} \) and \( \textbf{D}^c \), respectively, and \( f_2(\infty) = 0 \), then it follows that \( f_1(z) = f(z)_+ \) and \( f_2(z) = f(z)_- \).

Define the pairing:
\[
\langle \omega(z), \xi(z) \rangle := \frac{1}{2\pi\mathrm{i}} \sum_{s=1}^{m} \int_{\gamma_s} \omega(z) \xi(z) \, dz, \quad \omega(z) \in \mathcal{H}, \quad \xi(z) \in T_{\lambda(z)}M^{cKP}.
\]
This pairing induces a surjective map from \( \mathcal{H} \) to \( T^{\ast}_{\lambda(z)}M^{cKP} \), so that an element of \( \mathcal{H} \) can  be regarded as a cotangent vector at \( \lambda(z) \).

\begin{lem}\label{kdvdulmap}
	For any \( \xi(z) \in T_{\lambda(z)}M^{cKP} \), define  linear maps from \( \mathcal{H} \) to \( T_{\lambda(z)}M^{cKP} \) as follows:
	\begin{equation}\label{kdvdulmet}
		\eta^{\ast}(\omega(z)) = -(\omega(z))_{+}\lambda'(z) + (\omega(z)\lambda'(z))_{+}
	\end{equation}
	and
	\begin{equation}\label{kdvdulpro}
		C_{\xi(z)}(\omega(z)) = -(\omega(z)\xi(z))_{+}\lambda'(z) + (\omega(z)\lambda'(z))_{+}\xi(z).
	\end{equation}
	Then we have
	\[
	\langle \eta^{\ast}(\omega(z)), \xi(z) \rangle_{\eta} = \langle \omega(z), \xi(z) \rangle
	\]
	and
	\[
	c(\xi_{1}(z), \xi_{2}(z), \eta^{\ast}(\omega(z))) = \langle C_{\xi_{1}(z)}(\omega(z)), \xi_{2}(z) \rangle_{\eta}.
	\]
\end{lem}

\begin{proof}
From direct calculation, we derive
	\begin{align*}
		\langle \eta^{\ast}(\omega), \xi \rangle_{\eta} =& -\res_{\infty} \left( \frac{\omega_{-}\lambda' \xi}{\lambda'} - \frac{(\omega\lambda')_{-} \xi}{\lambda'} \right) dz - \sum_{s=1}^{m} \res_{a_{s,0}} \left( -\frac{\omega_{+}\lambda' \xi}{\lambda'} + \frac{(\omega\lambda')_{+} \xi}{\lambda'} \right) dz \\
		=& -\res_{\infty} \omega_{-} \xi \, dz + \sum_{s=1}^{m} \res_{a_{s,0}} \omega_{+} \xi \, dz \\
		=& \langle \omega, \xi \rangle.
	\end{align*}
	In a similar manner, we obtain
	\begin{align*}
		\langle C_{\xi_{1}(z)}(\omega), \xi_{2} \rangle_{\eta} =& -\res_{\infty} \left( \frac{(\omega\xi_{1})_{-}\lambda' \xi_{2}}{\lambda'} - \frac{(\omega\lambda')_{-} \xi_{1} \xi_{2}}{\lambda'} \right) dz \\
		& - \sum_{s=1}^{m} \res_{a_{s,0}} \left( -\frac{(\omega\xi_{1})_{+}\lambda' \xi_{2}}{\lambda'} + \frac{(\omega\lambda')_{+} \xi_{1} \xi_{2}}{\lambda'} \right) dz \\
		=& -\res_{\infty} \omega \xi_{1} \xi_{2} \, dz + \res_{\infty} \frac{(\omega\lambda')_{-} \xi_{1} \xi_{2}}{\lambda'} dz - \sum_{s=1}^{m} \res_{a_{s,0}} \frac{(\omega\lambda')_{+} \xi_{1} \xi_{2}}{\lambda'} dz \\
		=& -\res_{\infty} \left( \frac{\omega_{-}\lambda' \xi_{1} \xi_{2}}{\lambda'} - \frac{(\omega\lambda')_{-} \xi_{1} \xi_{2}}{\lambda'} \right) dz \\
		& - \sum_{s=1}^{m} \res_{a_{s,0}} \left( -\frac{\omega_{+}\lambda' \xi_{1} \xi_{2}}{\lambda'} + \frac{(\omega\lambda')_{+} \xi_{1} \xi_{2}}{\lambda'} \right) dz \\
		=& c(\xi_{1}, \xi_{2}, \eta^{\ast}(\omega)).
	\end{align*}
	The lemma is proved.
\end{proof}

\subsection{Hamiltonian structure}
Let \( LM^{cKP} \) denote the loop space of \( M^{cKP}\). According to the Dubrovin-Novikov theorem, the Hamiltonian structure \( \mathcal{P} \) on \( LM^{cKP} \) corresponding to the flat metric \( \langle\ ,\ \rangle_{\eta} \)  is given by
\[
\mathcal{P}(\omega) = \eta^{\ast} \cdot \nabla_{\p_{x}} \omega.
\]
\begin{lem}\label{kdvhamthm}
	The Hamiltonian operator $\mathcal{P}$ has the explicit form
		\begin{equation}\label{kdvham}
		\mathcal{P}(\omega(z))=\{\omega(z)_{+},\lambda(z)\}-\{\omega(z),\lambda(z)\}_{+},
	\end{equation}
	where
	$$
	\{f(z,x),g(z,x)\}=\frac{\p f(z,x)}{\p z}\frac{\p g(z,x)}{\p x}-\frac{\p f(z,x)}{\p x}\frac{\p g(z,x)}{\p z}.
	$$
\end{lem}
\begin{proof}
	We now deduce the explict form of the operator $\eta^{\ast}\cdot\nabla_{\p}\omega$ from the equality
	$$
	\langle\eta^{\ast}\cdot\nabla_{\p}\omega,\xi\rangle_{\eta}+\langle\omega,\nabla_{\p}\xi\rangle=\p \langle\omega,\xi\rangle.
	$$
	By direct computation, we have
	$$
	\p \langle\omega,\xi\rangle=\frac{1}{2\pi\mathrm{i}}\dsum_{s=1}^{m}\int_{\gamma_{s}}(\p\omega\xi+\omega\p\xi)dz,
	$$
	and
	\begin{align*}
		\langle\omega,\nabla_{\p}\xi\rangle=&\frac{1}{2\pi\mathrm{i}}\dsum_{s=1}^{m}\int_{\gamma_{s}}(\omega(\p\xi-(\frac{\xi\p\lambda}{\lambda'})'_{\infty,\ge 0}-\dsum_{s'=1}^{m}(\frac{\xi\p\lambda}{\lambda'})'_{a_{s',0},\le 0})dz\\
		=&\frac{1}{2\pi\mathrm{i}}\dsum_{s=1}^{m}\int_{\gamma_{s}}(\omega\p\xi-\omega_{-}(\frac{\xi\p\lambda}{\lambda'})'_{\infty,\ge 0}-\omega_{+}\dsum_{s'=1}^{m}(\frac{\xi\p\lambda}{\lambda'})'_{a_{s',0},\le 0})dz\\
		=&\frac{1}{2\pi\mathrm{i}}\dsum_{s=1}^{m}\int_{\gamma_{s}}\omega\p\xi dz-\res_{\infty}\frac{\omega'_{-}\xi\p \lambda}{\lambda'}dz+\dsum_{s'=1}^{m}\res_{a_{s',0}}\frac{\omega'_{+}\xi\p \lambda}{\lambda'}dz.
	\end{align*}
	Hence
	\begin{align*}
		\eta^{\ast}\cdot\nabla_{\p}\omega=&\omega'_{+}\p\lambda-\p \omega_{+}\lambda'-(\omega'\p\lambda-\p\omega\lambda')_{+}\\
		=&-\omega'_{-}\p\lambda+\p \omega_{-}\lambda'+(\omega'\p\lambda-\p\omega\lambda')_{-}.
	\end{align*}
	Setting $\p=\p_{x}$, we obtain the equality \eqref{kdvham}.
\end{proof}
\subsection{principal hierarchy for $M^{cKP}$}\label{kdvpriproof}
To prove Theorem \ref{maincKP}, we first need to establish the following lemma.

\begin{lem}
	Let \( Q_{p}(\lambda), p \in \mathbb{N} \), be analytic functions in $\lambda$  that satisfy \[\frac{\partial Q_{p}(\lambda)}{\partial \lambda}  = Q_{p-1}(\lambda). \] 
	Define
	\begin{equation}\label{kdvFdef}
		F_{i,p} = \frac{1}{2\pi i} \int_{\gamma_{i}} Q_{p+1}(\lambda(z)) dz,
	\end{equation}
	then we have
	\begin{equation}\label{Fcond}
		\eta^{\ast} \cdot \nabla_{\partial} dF_{i,p} = C_{\partial}(dF_{i,p-1}),
	\end{equation}
	where $\p$ is any vector field on $M^{cKP}$, and the operators \( \eta^{\ast} \cdot \nabla_{\partial} \) and \( C_{\partial} \) are given by equalities \eqref{kdvdulmet} and \eqref{kdvdulpro}, respectively. 
\end{lem}

	\begin{proof}
		The differential of $F_{i,p}$ at $\lambda(z)\in M^{cKP}$ is given by
		$$
		dF_{i,p}=\frac{\p Q_{p+1}(\lambda)}{\p \lambda}\textbf{1}_{\gamma_{i}}\in \mathcal{H},
		$$
			where the functions \(\{\mathbf{1}_{\gamma_i}\}_{i=1}^m\) belonging to the space \(\mathcal{H}\) are characterized  by
		\[
		\mathbf{1}_{\gamma_i}|_{\gamma_j} = \delta_{ij}, \quad  i, j = 1, \ldots, m.
		\]
		By the identity
		$$
		\{Q_{p}(\lambda(z)),\lambda(z)\}=0,
		$$
		where $\{f,g\}=f' \p g-g' \p f$, we obtain
		$$
		\eta^{\ast}\cdot\nabla_{\p}dF_{i,p}=\{(Q_{p}(\lambda(z))\textbf{1}_{\gamma_{i}})_{+},\lambda(z)\}.
		$$
		For the right-hand side of equality \eqref{Fcond}, we have
		\begin{align*}
			C_{\p}(dF_{i,p-1})=&-(\frac{\p Q_{p}(\lambda)}{\p\lambda}\p\lambda(z)\textbf{1}_{\gamma_{i}})_{+}\lambda'(z) + (\frac{\p Q_{p}(\lambda)}{\p\lambda}\lambda'(z)\textbf{1}_{\gamma_{i}})_{+}\p\lambda(z)\\
			=&\{(Q_{p}(\lambda(z))\textbf{1}_{\gamma_{i}})_{+},\lambda(z)\})\\
			=&	\eta^{\ast}\cdot\nabla_{\p}dF_{i,p}.
		\end{align*}
		Thus, the lemma is proved.
	\end{proof}

\begin{proof}[Proof of Theorem \ref{maincKP}]
	For the Hamiltonian density \( \theta_{t_{0,j},p} \), let \( M' \) be a subset of \( M^{cKP} \) such that \( w_{0} = \lambda^{\frac{1}{n_{0}}} \) can be analytically continued onto \( \cup_{s=1}^{m} \gamma_{s} \), that is, \( w_{0} \in \mathcal{H} \). Then on \( M' \), \( \theta_{t_{0,j},p} \) can be expressed as a smooth function of the form \eqref{kdvFdef}, and thus satisfies equation \eqref{princon1}. The uniqueness of analytic function implies that \( \theta_{t_{0,j},p} \) satisfies equation \eqref{princon1} on \( M^{cKP} \).
	
	By a similar method, it can be shown that the remaining Hamiltonian densities \( \theta_{t_{i,j},p} \) also satisfy equation \eqref{princon1}. In particular, for \( \theta_{t_{i,n_{i}},p} \), consider a subset \( M' \) of \( M^{cKP} \) such that \( w_{0} = \lambda^{\frac{1}{n_{0}}} \) can be analytically continued onto \( \cup_{s=1}^{m} \gamma_{s} \), with the winding number \( 1 \) along \( \gamma_{i} \) and \( 0 \) along \( \gamma_{j} \) for \( j \neq i \), and \( w_{i} = \lambda^{\frac{1}{n_{i}}} \) can be analytically continued onto \( \gamma_{i} \) with the winding number \( -1 \). Then on \( M' \), we have
	\begin{align*}
		\theta_{t_{i,n_{i}},p} =& \res_{\infty} \frac{c_{p}}{n_{0}} \frac{\lambda^{p}}{p!} dz + \frac{1}{2\pi\mathrm{i}} \int_{\gamma_{i}} \frac{\lambda^{p}}{p!} (\log w_{0} w_{i} - \frac{c_{p}}{n_{i}}) dz + \frac{1}{2\pi\mathrm{i}} \sum_{s \neq i} \int_{\gamma_{s}} \frac{\lambda^{p}}{p!} \log w_{0} dz\\
		&= \frac{1}{2\pi\mathrm{i}} \int_{\gamma_{i}} \frac{\lambda^{p}}{p!} (\log w_{0} w_{i} - \frac{c_{p}}{n_{i}} - \frac{c_{p}}{n_{0}}) dz + \frac{1}{2\pi\mathrm{i}} \sum_{s \neq i} \int_{\gamma_{s}} \frac{\lambda^{p}}{p!} (\log w_{0} - \frac{c_{p}}{n_{0}}) dz,
	\end{align*}
	which thus satisfies equation \eqref{princon1}.
	
	Next, we verify that \( \theta_{t_{i,j},p} \) satisfies equation \eqref{princon2}. Introduce the operator \( \mathcal{E} = E + \frac{1}{n_{0}}z\frac{\partial}{\partial z} \), then we have
	\[
	\operatorname{Lie}_{\mathcal{E}}\lambda(z) = \lambda(z)
	\]
	and
	\[
	\res_{a_{i,0}}\operatorname{Lie}_{\mathcal{E}} f(\lambda(z)) \, dz = \operatorname{Lie}_E \res_{a_{i,0}} f(\lambda(z)) \, dz - \frac{1}{n_{0}} \res_{a_{i,0}} f(\lambda(z)) \, dz,
	\]
	which implies that
	\[
	\operatorname{Lie}_E \theta_{u, p}(t) = \begin{cases}
		\left(p + 1 - \frac{j}{n_{0}} + \frac{1}{n_{0}}\right) \theta_{u, p}(t), & u = t_{0,j}; \\
		\left(p + 1 - \frac{j}{n_{i}} + \frac{1}{n_{0}}\right) \theta_{u, p}(t), & u = t_{i,j},\ j \neq n_{i}; \\
		\left(p + \frac{1}{n_{0}}\right) \theta_{u, p}(t) + \sum_{s=1}^{m} \frac{1}{n_{0}}\theta_{t_{s,0},p-1} + \frac{1}{n_{i}}\theta_{t_{i,0},p-1}, & u = t_{i,n_{i}}.
	\end{cases}
	\]
	Hence, equation \eqref{princon2} is satisfied.
	
	We now deduce the Hamiltonian vector fields corresponding to the densities $\{ \theta_{t_{i,j},p} \}$. For the density $\theta_{t_{0,j},p}$, 
	we have
	\begin{align*}
			\mathcal{P}(d\theta_{t_{0,j},p}) =&\{(c_{0,j;p-1}w_{0}^{pn_{0}-j})_{+}, \lambda\}\\
			=&\{(c_{0,j;p-1}w_{0}^{pn_{0}-j})_{\infty, \ge 0}, \lambda\}.
	\end{align*}
	For the case of $\theta_{t_{i,j},p}$ where $j\ne n_{i}$, we have
	\begin{align*}
		\mathcal{P}(d\theta_{t_{i,j},p}) =& -\{(c_{i,j;p-1}w_{i}^{pn_{i}-j} \mathbf{1}_{\gamma_{i}})_{-}, \lambda\} \\
		=& -\{(c_{i,j;p-1}w_{i}^{pn_{i}-j})_{a_{i,0}, \le -1}, \lambda\}.
	\end{align*}
	For the specific density $\theta_{t_{i,n_{i}},p}$, we have
	\begin{align*}
		\mathcal{P}(d\theta_{t_{i,n_{i}},p+1})=&\{(\frac{\lambda^{p}}{p!}(\log w_{0}w_{i}-\frac{c_{p}}{n_{i}}-\frac{c_{p}}{n_{0}})\textbf{1}_{\gamma_{i}})_{+},\lambda\}+\dsum_{s\ne i}\{(\frac{\lambda^{p}}{p!}(\log w_{0}-\frac{c_{p}}{n_{0}})\textbf{1}_{\gamma_{s}})_{+},\lambda\}\\
		=&\{(\frac{\lambda^{p}}{p!}(\log \frac{w_{0}}{z-a_{i,0}}+\log w_{i}(z-a_{i,0})-\frac{c_{p}}{n_{i}}-\frac{c_{p}}{n_{0}})\textbf{1}_{\gamma_{i}})_{+},\lambda\}\\
		&+\dsum_{s\ne i}\{(\frac{\lambda^{p}}{p!}(\log \frac{w_{0}}{z-a_{i,0}}+\log(z-a_{i,0})-\frac{c_{p}}{n_{0}})\textbf{1}_{\gamma_{s}})_{+},\lambda\}\\
		=&\{(\frac{\lambda^{p}}{p!}(\log\frac{w_{0}}{z-a_{i,0}}-\frac{c_{p}}{n_{0}}))_{\infty,\ge 0},\lambda\}-\{(\frac{\lambda^{p}}{p!}(\log w_{i}(z-a_{i,0})-\frac{c_{p}}{n_{i}}))_{a_{i,0},\le -1},\lambda\}\\
		&+\{\frac{\lambda^{p}}{p!}(\log w_{i}(z-a_{i,0})-\frac{c_{p}}{n_{i}})\textbf{1}_{\gamma_{i}},\lambda\}+\dsum_{s\ne i}\{(\log(z-a_{i,0})\textbf{1}_{\gamma_{s}})_{+},\lambda\}\\
		=&\{\frac{\lambda^{p}}{p!}((\log\frac{w_{0}}{z-a_{i,0}}-\frac{c_{p}}{n_{0}}))_{\infty,\ge 0},\lambda\}-\{(\frac{\lambda^{p}}{p!}(\log(z-a_{i,0})w_{i}-\frac{c_{p}}{n_{i}}))_{a_{i,0},\le -1},\lambda\}\\
		&-\sum_{s\ne i}\{(\frac{\lambda^{p}}{p!}\log(z-a_{i,0}))_{a_{s,0},\le -1},\lambda\}+\{\frac{\lambda^{p}}{p!}\log(z-a_{i,0}),\lambda\}.
	\end{align*}
    Thus, the theorem is proved.
	\end{proof}
	
	\subsection{principal hierarchy for $M^{D-cKP}$}
	Let $M^{D-cKP}$ be the submanifold of $M^{cKP}$ consisting of elements of the form \eqref{Dckpsup}. Let us first show that $M^{D-cKP}$ is a natural Frobenius submanifold of $M^{cKP}$, as defined by Strachan \cite{STRACHAN2001}.
	
	For $\lambda(z)\in M^{D-cKP}$, we have
	\begin{align*}
		&w_{0}(-z)=-w_{0}(z),\quad z\to \infty;\\
		&w_{1}(-z)=-w_{1}(z),\quad z\to 0;\\
		&w_{2i-2}(-z)=w_{2i-1}(z),\quad p\to b_{2i-1,0},\ i=2,3,\cdots,m',
	\end{align*}
	where $b_{2i-1,0}=-b_{2i-2,0}$. Hence, for $p\to\infty$, we have 
	\begin{align*}
		-z=&w_{0}(-z)-t_{0,1}w_{0}^{-1}(-z)-t_{0,2}w_{0}^{-2}(-z)-\cdots\\
		=&-w_{0}(z)+t_{0,1}w_{0}^{-1}(z)-t_{0,2}w_{0}^{-2}(z)-\cdots\\
		=&-w_{0}(z)+t_{0,1}w_{0}^{-1}(z)+t_{0,2}w_{0}^{-2}(z)+\cdots.
	\end{align*}
	For $p\to 0$, we have
	\begin{align*}
		-z=&t_{1,0}+t_{1,1}w_{1}^{-1}(-z)+t_{1,2}w_{1}^{-2}(-z)+\cdots\\
		=&t_{1,1}w_{1}^{-1}(-z)+t_{1,2}w_{1}^{-2}(-z)+\cdots\\
		=&-t_{1,1}w_{1}^{-1}(z)+t_{1,2}w_{1}^{-2}(z)-\cdots\\
		=&-t_{1,1}w_{1}^{-1}(z)-t_{1,2}w_{1}^{-2}(z)+\cdots.
	\end{align*}
	For $p\to b_{2i-1,0},\ i=2,\cdots,m',$ we have
	\begin{align*}
		-z=&t_{2i-1,0}+t_{2i-1,1}w_{2i-1}^{-1}(-z)+t_{2i-1,2}w_{2i-1}^{-2}(-z)+\cdots\\
		=&t_{2i-1,0}+t_{2i-1,1}w_{2i-2}^{-1}(z)+t_{2i-1,2}w_{2i-2}^{-2}(z)+\cdots\\
		=&-t_{2i-2,0}-t_{2i-2,1}w_{2i-2}^{-1}(z)-t_{2i-2,2}w_{2i-2}^{-2}(z)-\cdots.
	\end{align*}
	Thus, we obtain the following restrictions for the flat coordinates:
	\begin{align*}
		&t_{0,2}=t_{0,4}=\cdots=t_{0,2n_{0}'-2}=0;\\
		&t_{1,0}=t_{1,2}=\cdots=t_{1,2n_{1}'}=0;\\
		&t_{2i-2,j}=-t_{2i-1,j},\quad i=2,\cdots,m',\ j=0,\cdots,n_{i}'.
	\end{align*}
	Hence, $M^{D-cKP}$ is a flat submanifold of $M^{cKP}$ with flat coordinates
	$$
	\textbf{t}^{D-cKP}=\{t_{0,2j-1}\}_{j=1}^{n_{0}'}\cup\{t_{1,2j-1}\}_{j=1}^{n_{1}'}\cup\{t_{2,j}\}_{j=0}^{n_{2}'}\cup\{t_{4,j}\}_{j=0}^{n_{3}'}\cup\cdots\cup\{t_{2m'-2,j}\}_{j=0}^{n_{m'}'}.
	$$
	
	On the other hand, let $p_{1},\cdots,p_{r}$ be the simple critical points of $\lambda(z)\in M^{cKP}$, where $r=dim(M^{cKP})$.  . For $\lambda(z)\in M^{D-cKP}$, we can choose $$
	p_{2j}=-p_{2j-1},\quad j=1,2,\cdots,\frac{r}{2},
	$$
	which implies
	$$
	u_{2j}=u_{2j-1},\quad j=1,2,\cdots,\frac{r}{2}.
	$$
	Thus, $M^{D-cKP}$ is a caustic submanifold of $M^{cKP}$ (for the definition, see reference \cite{STRACHAN2004}). According to Corollary 3.7 in \cite{STRACHAN2004}, $M^{D-cKP}$ is a natural Frobenius submanifold of $M^{cKP}$.
	
	\begin{rem}
		Observe that the above conclusion also holds for the case \( n_1' = 0 \). In this case, we have
		\[
		r = \dim (M^{cKP}) = 2\dim (M^{D-cKP}) - 1.
		\]
		Choose simple critical points \( p_1, \ldots, p_r \) such that
		\[
		p_1 = 0, \quad p_{2j} = -p_{2j+1}, \quad j = 1, 2, \ldots, \frac{r-1}{2},
		\]
		then we have
		\[
		u_{2j} = u_{2j+1}, \quad j = 1, 2, \ldots, \frac{r-1}{2}.
		\]
	\end{rem}
	We now consider the principal hierarchy for $M^{D-cKP}$. Let \( \mathcal{H}^{odd} \) be the space consisting of elements in \( \mathcal{H} \) that satisfy the condition \( \omega(-z) = -\omega(z) \). For \( \omega(z) \in \mathcal{H}^{odd} \), we have
	\[
	\omega_{-}(-z) = \begin{cases}
		-\omega_{-}(z), & \text{if } z, -z \in D^{c}, \\
		\omega_{+}(z), & \text{if } -z \in D^{c}, z \in D,
	\end{cases}
	\]
	and
	\[
	\omega_{+}(-z) = \begin{cases}
		\omega_{+}(z), & \text{if } z, -z \in D, \\
		-\omega_{-}(z), & \text{if } -z \in D, z \in D^{c}.
	\end{cases}
	\]
	Thus, for \( \lambda(z) \in M^{D-cKP} \), we have \( \eta^{\ast}(\omega(z)) \in T_{\lambda(z)}M^{D-cKP} \), and
	\[
	\eta^{\ast} \cdot \nabla_{\p}\omega \in T_{\lambda(z)}M^{D-cKP}
	\]
	for any \( \p \in T_{\lambda(z)}M^{D-cKP} \).
We observe that the differentials of the Hamiltonian densities
	\begin{align*}
	&\theta_{t_{0,2j-1},p},\quad j=1,2,\cdots,n_{0}';\\
	&\theta_{t_{1,2j-1},p},\quad j=1,2,\cdots,n_{1}';
\end{align*}
and	
$$
\theta_{t_{2i-2,j},p}-\theta_{t_{2i-1,j},p},\quad i=2,\cdots,m',\ j=0,\cdots,n_{i}' 
$$
belong to \(\mathcal{H}^{odd}\), hence the corresponding Hamiltonian vector fields  form the principal hierarchy for $M^{D-cKP}$.

\subsection{principal hierarchy for $\hat{M}^{cKP}$}
Let \( \hat{M}^{cKP} \) be the almost duality of the Frobenius manifold \( M^{cKP} \). We now construct the principal hierarchy for \( \hat{M}^{cKP} \). 

Suppose $I\in \textbf{D}$, where \( I \) is the set of all zeros and poles of \( \lambda(z) \) in  \( \mathbb{C} \). Define the linear map \( g^{\ast}(\omega(z)) = C_{E}(\omega(z)) \) from \( \mathcal{H} \) to \( T_{\lambda(z)}M \) as:
\[
g^{\ast}(\omega(z)) = (\lambda(z)\omega(z))_{-}\lambda'(z) - \lambda(z)(\lambda'(z)\omega(z))_{-} + \frac{\lambda'(z)}{n_{0}}\frac{1}{2\pi\mathrm{i}} \sum_{s=1}^{m} \int_{\gamma_{s}} \lambda'(z)\omega(z) \, dz.
\]
We have
\[
\langle \omega(z), \xi(z) \rangle = (g^{\ast}(\omega(z)), \xi(z))_{g},
\]
where
\[
(\p_{1}, \p_{2})_{g} = \sum \res_{d\lambda=0} \frac{\p_{1}\log\lambda(z) \p_{2}\log\lambda(z)}{(\log\lambda(z))'} \, dz
\]
is the intersection form on $M^{cKP}$. By the following lemma, we derive the explicit form of the Hamiltonian structure $\hat{\mathcal{P}}$ for the flat metric $g$. 
\begin{lem}\label{cKPdulcon}
	Let \( \hat{\nabla} \) be the Levi-Civita connection associated with \( g \).  Then, 
	\[
	\hat{\nabla}_{\p_{1}}\p_{2} \cdot \lambda(z) = \p_{1}\p_{2}\lambda(z) - \frac{\p_{1}\lambda(z)\p_{2}\lambda(z)}{\lambda(z)} - \lambda(z) \sum_{q \in I} \left( \frac{\p_{1}\lambda(z)\p_{2}\lambda(z)}{\lambda(z)\lambda'(z)} \right)_{q,\le-1}'.
	\]
 Morever, we have
	\begin{equation}\label{cKPdulcome}
			g^{\ast} \cdot \hat{\nabla}_{\p}\omega(z) = \{\omega(z),\lambda(z)\}_{-}\lambda(z) - \{(\omega(z)\lambda(z))_{-},\lambda(z)\} - \frac{\lambda'(z)}{n_{0}}\frac{1}{2\pi\mathrm{i}} \sum_{s=1}^{m} \int_{\gamma_{s}} \{\omega(z),\lambda(z)\} \, dz,
	\end{equation}
	where \( \{f,g\} = f'\p g - g'\p f \).
\end{lem}
\begin{proof}
	The compatibility of $\hat{\nabla}$ with $g$ can be verified as follows:
	\begin{align*}
		&(\hat{\nabla}_{\p_{1}}\p_{2},\p_{3})_{g}+	(\p_{2},\hat{\nabla}_{\p_{1}}\p_{3})_{g}\\
		=&(-\res_{\infty}-\sum_{q\in I}\res_{q})(\frac{\p_{1}\p_{2}\lambda(z)\p_{3}\lambda(z)}{\lambda'(z)\lambda(z)}dz+\frac{\p_{2}\lambda(z)\p_{1}\p_{3}\lambda(z)}{\lambda'(z)\lambda(z)}dz)+K,
	\end{align*}
	where
	\begin{align*}
		K=&(-\res_{\infty}-\sum_{q\in I}\res_{q})(-\frac{\p_{1}\lambda(z)\p_{2}\lambda(z)\p_{3}\lambda(z)dz}{\lambda'(z)\lambda(z)^{2}}-(\frac{\p_{1}\lambda(z)\p_{2}\lambda(z)}{\lambda(z)\lambda'(z)})'\frac{\p_{3}\lambda(z)dz}{\lambda'(z)}+c.p.\{2,3\})\\
		=&(-\res_{\infty}-\sum_{q\in I}\res_{q})((\frac{\p_{1}\lambda(z)\p_{2}\lambda(z)\p_{3}\lambda(z)dz}{\lambda'(z)\lambda'(z)})(\frac{1}{\lambda(z)})'+\frac{\p_{1}\lambda(z)\p_{2}\lambda(z)dz}{\lambda(z)\lambda'(z)}(\frac{\p_{3}\lambda(z)}{\lambda'(z)})'+c.p.\{2,3\})\\
		=&(-\res_{\infty}-\sum_{q\in I}\res_{q})((\frac{\p_{1}\lambda(z)\p_{2}\lambda(z)dz}{\lambda'(z)})(\frac{\p_{3}\lambda(z)}{\lambda'(z)\lambda(z)})'+c.p.\{2,3\})\\
		=&(-\res_{\infty}-\sum_{q\in I}\res_{q})\lambda(z)\p_{1}\lambda(z)(\frac{\p_{2}\lambda(z)\p_{3}\lambda(z)}{(\lambda'(z)\lambda(z))^{2}})'dz\\
		=&(-\res_{\infty}-\sum_{q\in I}\res_{q})\p_{2}\lambda(z)\p_{3}\lambda(z)\p_{1}(\frac{1}{\lambda'(z)\lambda(z)})dz.
	\end{align*}
	This implies that
	$$
	(\hat{\nabla}_{\p_{1}}\p_{2},\p_{3})_{g}+	(\p_{2},\hat{\nabla}_{\p_{1}}\p_{3})_{g}=\hat{\nabla}_{\p_{1}}(\p_{2},\p_{3})_{g}.
	$$
	On the other hand, 
	\begin{align*}
		\langle \omega, \hat{\nabla}_{\p}\xi\rangle=&\sum_{s=1}^{m}\int_{\gamma_{s}}(\omega\p\xi -\frac{\omega\xi\p\lambda}{\lambda}+(\omega\lambda)' \sum_{q\in I}(\frac{\xi \p\lambda}{\lambda\lambda'})_{q,\le -1})dz\\
		&=\sum_{s=1}^{m}\int_{\gamma_{s}}(\omega\p\xi -\frac{(\omega\p\lambda)_{+}\xi}{\lambda}+(\omega\lambda)'_{+} \sum_{q\in I}(\frac{\xi \p\lambda}{\lambda\lambda'})_{q,\le -1})dz\\
		&=\sum_{s=1}^{m}\int_{\gamma_{s}}(\omega\p\xi -\frac{(\omega\p\lambda)_{+}\xi}{\lambda}+(\omega\lambda)'_{+} \frac{\xi \p\lambda}{\lambda\lambda'})dz,
	\end{align*}
	and
		\begin{align*}
		(g^{\ast}\cdot\hat{\nabla}_{\p}\omega,\xi)_g=&\sum_{q\in I}\res_{q}(\frac{(\omega'\p\lambda-\lambda'\p\omega)_{+}\xi}{\lambda'}-\frac{(\omega\lambda')_{+}\xi\p\lambda}{\lambda\lambda'}+\frac{(\lambda\p\omega+\omega\p\lambda)_{+}\xi}{\lambda})dz\\
		=&\sum_{s=1}^{m}\int_{\gamma_{s}}(-\frac{(\omega\lambda')_{+}\xi\p\lambda}{\lambda\lambda'}+\frac{(\lambda\p\omega+\omega\p\lambda)_{+}\xi}{\lambda})dz\\
		=&\sum_{s=1}^{m}\int_{\gamma_{s}}(-\frac{(\omega\lambda')_{+}\xi\p\lambda}{\lambda\lambda'}+\frac{(\omega\p\lambda)_{+}\xi}{\lambda}+\xi\p\omega)dz.
	\end{align*}
	Hence
			$$
	( g^{\ast}\cdot\hat{\nabla}_{\p}\omega,\xi)_g+\langle\omega,\hat{\nabla}_{\p}\xi\rangle=\p \langle\omega,\xi\rangle.
	$$
	The lemma is proved.
\end{proof}
 
\begin{proof}[Proof of Theorem \ref{maincKPdul}]
	Let \( F_{\gamma_{j},p} = \frac{1}{2\pi \mathrm{i}} \int_{\gamma_{j}} \tilde{Q}_{p}(\lambda) \, dz \), where \( Q_{p}(\lambda) = \frac{\partial \tilde{Q}_{p}(\lambda)}{\partial \lambda} \) satisfies the recurrence relation:
	\[
	\lambda \frac{\partial Q_{p+1}}{\partial \lambda} + Q_{p+1} = Q_{p}.
	\]
	Then we have \( dF_{p} = Q_{p} \mathbf{1}_{\gamma_{j}} \), and
	\begin{align*}
		g^{\ast} \cdot \hat{\nabla}_{\p} dF_{p+1} =& - ((\frac{\partial Q_{p+1}}{\partial \lambda} \lambda \lambda' \mathbf{1}_{\gamma_{j}} + Q_{p+1} \lambda' \mathbf{1}_{\gamma_{j}})_{-} \p \lambda) + ((\frac{\partial Q_{p+1}}{\partial \lambda} \lambda \p \lambda \mathbf{1}_{\gamma_{j}} + Q_{p+1} \p \lambda \mathbf{1}_{\gamma_{j}})_{-} \lambda') \\
		=& (Q_{p} \p \lambda \mathbf{1}_{\gamma_{j}})_{-} \lambda' - (Q_{p} \lambda' \mathbf{1}_{\gamma_{j}})_{-} \p \lambda \\
		=& C_{\p}(dF_{p}).
	\end{align*}
	
	Consider a subset \( \hat{M}' \) of \( \hat{M}^{cKP} \) such that for any \( \lambda(z) \in \hat{M}'  \), the winding number of \( \lambda(z) \) along \( \gamma_{j} \) is zero. Let \( \tilde{Q}(p) = \frac{(\log(\lambda(z)))^{p+1}}{(p+1)!} \), then 
		$$
	F_{\gamma_{j},p} = \frac{1}{2\pi\mathrm{i}} \int_{\gamma} \frac{(\log(\lambda(z)))^{p+1}}{(p+1)!} \, dz
	$$
	 satisfies equality \eqref{princon1} on \( \hat{M}'  \). Furthermore, by the uniqueness of analytic function, for any \( \tilde{\gamma} \in [\gamma_{j}] \), \( F_{\tilde{\gamma},p} \) satisfies equality \eqref{princon1} on $\hat{M}^{cKP}$, where \( [\gamma_{j}] \) denotes the homotopy equivalence class of \( \gamma_{j} \) in \( \mathbb{C}-I \).
	
	In particular, let $q_{1}, \cdots$ and $q_{k}, p_{1}, \cdots, p_{k}$ be the zeros and poles of \(\lambda(z)\) within the region surrounded by $\tilde{\gamma}$,
	respectively, then
	\begin{align*}
		F_{\tilde{\gamma},0} =& \frac{1}{2\pi\mathrm{i}} \int_{\tilde{\gamma}} \log(\lambda(z)) \, dz \\
		=& \frac{1}{2\pi\mathrm{i}} \int_{\tilde{\gamma}} \log(\lambda(z) \prod_{s=1}^{k} \frac{z-p_{s}}{z-q_{s}}) \, dz + \frac{1}{2\pi\mathrm{i}} \int_{\tilde{\gamma}} \log(\prod_{s=1}^{k} \frac{z-q_{s}}{z-p_{s}}) \, dz \\
		=& -\res_{\infty} \log(\prod_{s=1}^{k} \frac{z-q_{s}}{z-p_{s}}) \, dz \\
		=& \sum_{s=1}^{k} p_{s} - \sum_{s=1}^{k} q_{s}.
	\end{align*}
	This completes the proof of the theorem.
\end{proof}
\begin{cor}
	The Hamiltonian vector field associated with the Hamiltonian density \( F_{\gamma_{j},p}\) for the Hamiltonian structure $\hat{\mathcal{P}}$ takes the form:
	$$
	\frac{\p\lambda(z)}{\p \hat{T}^{j,p}}=\{(\frac{(\log(\lambda(z)))^{p}}{p!}\textbf{1}_{\gamma_{j}})_{+},\lambda(z)\},
	$$
		where
	$
	\{f,g\}=f'\p_{x}g-g'\p_{x}f.
	$
\end{cor}
\subsection{rank-1 extension}\label{extcKP}

To construct a rank-1 extension of the Frobenius manifold $M^{cKP}$, we  need the following lemma
\begin{lem}\label{kdvopenlem}
	For any vector fields $\p_{1},\p_{2}$ on $M^{cKP}$, we have
$$
\frac{\p_{1}\lambda(z)\p_{2}\lambda(z)-\p_{1}\lambda(z)\circ \p_{2}\lambda(z)}{\lambda'}=(\frac{\p_{1}\lambda(z)\p_{2}\lambda(z)}{\lambda'(z)})_{\infty,\ge 0}+\sum_{j=1}^{m}(\frac{\p_{1}\lambda(z)\p_{2}\lambda(z)}{\lambda'(z)})_{\varphi_{j},\le -1}.
$$
\end{lem}
\begin{proof}
Assume without loss of generality that $\p_{1}=\p_{u_{i}},\p_{2}=\p_{u_{j}}$, where $\{u_{i}\}_{i=1}^{n}$ are the canonical coordinates on $M^{cKP}$. Using the property
	$$
	\p_{u_{i}} \lambda(z)|_{p_{j}}=\delta_{ij},\quad i,j=1,\cdots,n,
	$$
	we obtain
	\begin{align*}
		&\frac{\p_{1}\lambda(z)\p_{2}\lambda(z)-\p_{1}\lambda(z)\circ \p_{2}\lambda(z)}{\lambda'(z)}\\
		=&
		(\frac{\p_{1}\lambda(z)\p_{2}\lambda(z)-\p_{1}\lambda(z)\circ \p_{2}\lambda(z)}{\lambda'(z)})_{\infty,\ge 0}+\sum_{j=1}^{m}	(\frac{\p_{1}\lambda(z)\p_{2}\lambda(z)-\p_{1}\lambda(z)\circ \p_{2}\lambda(z)}{\lambda'(z)})_{\varphi_{j},\le -1}\\
		=&(\frac{\p_{1}\lambda(z)\p_{2}\lambda(z)}{\lambda'(z)})_{\infty,\ge 0}+\sum_{j=1}^{m}(\frac{\p_{1}\lambda(z)\p_{2}\lambda(z)}{\lambda'(z)})_{\varphi_{j},\le -1}.
	\end{align*}
\end{proof}
	Define  $\Omega(\textbf{t},s)$ such that
	$$
	\p_{s}\p_{\alpha}\Omega(\textbf{t},s)=\p_{\alpha}\lambda(s),\quad \p^{2}_{s}\Omega(\textbf{t},s)=\lambda'(s),
	$$
	and 
	$$ \p_{\alpha}\p_{\beta}\Omega(\textbf{t},s)=(\frac{\p_{\alpha}\lambda(s)\p_{\beta}\lambda(s)}{\lambda'(s)})_{\infty,\ge 0}+\sum_{j=1}^{m}(\frac{\p_{\alpha}\lambda(s)\p_{\beta}\lambda(s)}{\lambda'(s)})_{\varphi_{j},\le -1}.
	$$
	From Lemmas \ref{kdvopenlem} and \ref{kdvconexp}, we obtain that $\omega(\textbf{t},s)=\p_{s}\Omega(\textbf{t},s)$ satisfies the condition of Lemma \ref{consopen}, thus defining a flat F-manifold structure on $M^{cKP}\times\mathbb{C}$, with the multiplication of the form:
	\begin{align*}
		&(\p_{\alpha},0)\star(\p_{\beta},0)=(\p_{\alpha}\circ\p_{\beta},\p_{\alpha}\p_{\beta}\Omega(\textbf{t},s)\cdot\p_{s}),\\
		&(\p_{\alpha},0)\star(0,\p_{s})=(0,\p_{\alpha}\lambda(s)\cdot\p_{s}),\\
		&(0,\p_{s})\star(0,\p_{s})=(0,\lambda'(s)\cdot\p_{s}).
	\end{align*}
	Here, we denote a vector field on $M^{cKP}\times\mathbb{C}$ by $X=(\bar{X},X(s)\p_{s})$ , where $\bar{X}$ and $X(s)\p_{s}$ are its components  along $M^{cKP}$ and $\mathbb{C}$, respectively. Let $\nabla$ denote the flat connection on the tangent bundle of  $M^{cKP}\times\mathbb{C}$, with flat coordinates $\textbf{t}\cup\{s\}$. We aim to determine the vector fields
$$
\Theta_{\bullet,p}=(\bar{\Theta}_{\bullet,p},\Theta_{\bullet,p}(s)\p_{s}),\quad \bullet\in\textbf{t}\cup\{s\},
$$
satisfying
\begin{equation}\label{openreu1}
\nabla_{X}\Theta_{\bullet,p+1}=X\star \Theta_{\bullet,p}
\end{equation}
and
\begin{equation}\label{openreu2}
	\Theta_{t_{i,j},p}=(\frac{\p}{\p t^{i,j}},0),\quad \Theta_{s,0}=(0,\p_{s}).
\end{equation}
The left-hand side of \eqref{openreu1} has the form:
$$
(\nabla_{\bar{X}}\bar{\Theta}_{\bullet,p+1}+X(s)\nabla_{\p_{s}}\bar{\Theta}_{\bullet,p+1},\p_{\bar{X}}\Theta_{\bullet,p+1}(s)\p_{s}+X(s)\p_{s}\Theta_{\bullet,p+1}(s)\cdot \p_{s}).
$$
For the right-hand side, we have
\begin{align*}
	&(\bar{X},0)\star(\bar{\Theta}_{\bullet,p},0)=(\bar{X}\circ\bar{\Theta}_{\bullet,p},( (\frac{\p_{\bar{X}}\lambda(s)\p_{\bar{\Theta}_{\bullet,p}}\lambda(s)}{\lambda'(s)})_{\infty,\ge 0}+\sum_{j=1}^{m}(\frac{\p_{\bar{X}}\lambda(s)\p_{\bar{\Theta}_{\bullet,p}}\lambda(s)}{\lambda'(s)})_{\varphi_{j},\le -1})\p_{s}),\\
	&(\bar{X},0)\star(0,\Theta_{\bullet,p}(s)\p_{s})=(0,(\p_{\bar{X}}\lambda(s))\Theta_{\bullet,p}(s)\p_{s}),\\
	&(0,X(s)\p_{s})\star(\bar{\Theta}_{\bullet,p},0)=(0,(\p_{\bar{\Theta}_{\bullet,p}}\lambda(s))X(s)\p_{s}),\\
	&(0,X(s)\p_{s})\star(0,\Theta_{\bullet,p}(s)\p_{s})=(0,\lambda'(s)\Theta_{\bullet,p}X(s)\p_{s}).
\end{align*}
\begin{thm}
	The vector fields on $M^{cKP}\times\mathbb{C}$ of the form
	\begin{align*}
		&\Theta_{s,p}=(0,\frac{\lambda(s)^{p}}{p!}\p_{s})\\
		&\Theta_{t_{0,n_{0}-j},p}	=(\eta^{\ast}(d\theta_{t_{0,j},p}),(d\theta_{t_{0,j},p}|_{z=s})_{+}\p_{s}),\\
		&\Theta_{t_{i,n_{i}-j},p}=(\eta^{\ast}(d\theta_{t_{i,j},p}),-(d\theta_{t_{i,j},p}|_{z=s})_{-}\p_{s}),\quad i=1,\cdots,m.
	\end{align*}
	satisfy equations \eqref{openreu1} and \eqref{openreu2}.
\end{thm}

\begin{proof}
Consider the case of $\Theta_{t_{0,n_{0}-j},p}$
. The left-hand side of equation \eqref{openreu1} is:
	\begin{align*}
		&(\nabla_{\bar{X}}\eta^{\ast}(d\theta_{t_{0,j},p+1}),(\p_{\bar{X}}\lambda(s)\cdot d\theta_{t_{0,j},p}|_{z=s})_{+}\p_{s}+X(s)(\lambda'(s)\cdot d\theta_{t_{0,j},p}|_{z=s})_{+}\p_{s})\\
		=&(\bar{X}\circ \eta^{\ast}(d\theta_{t_{0,j},p}),(\p_{\bar{X}}\lambda(s)\cdot d\theta_{t_{0,j},p}|_{z=s})_{+}\p_{s}+X(s)(\lambda'(s)\cdot d\theta_{t_{0,j},p}|_{z=s})_{+}\p_{s}).
	\end{align*}
	The right-hand side of \eqref{openreu1} is:
	\begin{align*}
		(\bar{X}\circ \eta^{\ast}(d\theta_{t_{0,j},p}),\mathcal{I}),
	\end{align*}
where
	\begin{align*}
		\mathcal{I}=&\frac{\p_{\bar{X}}\lambda(s)\cdot \eta^{\ast}(d\theta_{t_{0,j},p})|_{z=s}-C_{\bar{X}}(d\theta_{t_{0,j},p})|_{z=s}}{\lambda'(s)}\p_{s}+\p_{\bar{X}}\lambda(s)\cdot(d\theta_{t_{0,j},p}|_{z=s})_{+}\p_{s}\\
		&+X(s)\eta^{\ast}(d\theta_{t_{0,j},p})|_{z=s}\p_{s}+X(s)\lambda'(s)\cdot(d\theta_{t_{0,j},p}|_{z=s})_{+}\p_{s}\\
		=&(\p_{\bar{X}}\lambda(s)\cdot d\theta_{t_{0,j},p}|_{z=s})_{+}\p_{s}-\p_{\bar{X}}\lambda(s)\cdot (d\theta_{t_{0,j},p}|_{z=s})_{+}\p_{s}+\p_{\bar{X}}\lambda(s)\cdot(d\theta_{t_{0,j},p}|_{z=s})_{+}\p_{s}\\
		&+X(s)(\lambda'(s)\cdot d\theta_{t_{0,j},p}|_{z=s})_{+}\p_{s}\\
		=&(\p_{\bar{X}}\lambda(s)\cdot d\theta_{t_{0,j},p}|_{z=s})_{+}\p_{s}+X(s)(\lambda'(s)\cdot d\theta_{t_{0,j},p}|_{z=s})_{+}\p_{s}.
	\end{align*}
	The remaining cases follow similarly.
\end{proof}
\begin{proof}[Proof of Theorem \ref{cKPopenpri}]
The principal hierarchy for the flat F-manifold $M^{cKP}\times \mathbb{C}$ is defined as:
$$
\frac{\p}{\p\tilde{T}^{\bullet,p}}=\Theta_{\bullet,p}\star\tilde{\p}_{x},
$$
where
$$
\tilde{\p}_{x}=(\p_{x},\frac{\p s}{\p x}\p_{s}).
$$
Consider the case of $\frac{\p}{\p\tilde{T}^{t_{0,n_{0}-j},p}}$. Through direct computation, we obtain:
$$
\frac{\p}{\p\tilde{T}^{t_{0,n_{0}-j},p}}=(\frac{\p}{\p T^{t_{0,j},p}},\mathcal{J})
$$
where
\begin{align*}
	\mathcal{J}=&\frac{\p_{x}\lambda(s)\cdot\eta^{\ast}(d\theta_{t_{0,j},p})|_{z=s}-C_{\p_{x}}(d\theta_{t_{0,j},p})|_{z=s}}{\lambda'(s)}\p_{s}+\p_{x}\lambda(s)\cdot(d\theta_{t_{0,j},p}|_{z=s})_{+}\p_{s}\\
	&+\frac{\p s}{\p x}\cdot\eta^{\ast}(d\theta_{t_{0,j},p})|_{z=s}\p_{s}+\lambda'(s)\cdot \frac{\p s}{\p x}\cdot (d\theta_{t_{0,j},p}|_{z=s})_{+}\p_{s}\\
	=&\frac{d_{x}\lambda(s)\cdot\eta^{\ast}(d\theta_{t_{0,j},p})|_{z=s}-C_{\p_{x}}(d\theta_{t_{0,j},p})|_{z=s}}{\lambda'(s)}\p_{s}+d_{x}\lambda(s)\cdot (d\theta_{t_{0,j},p}|_{z=s})_{+}\p_{s}\\
	=&\frac{d_{x}\lambda(s)\cdot(d\theta_{t_{0,j},p}|_{z=s}\cdot \lambda'(s))_{+}}{\lambda'(s)}\p_{s}+(d\theta_{t_{0,j},p}|_{z=s}\cdot \p_{x}\lambda(s))_{+}\p_{s}-\frac{\p_{x}\lambda(s)\cdot(d\theta_{t_{0,j},p}|_{z=s}\cdot \lambda'(s))_{+}}{\lambda'(s)}\p_{s}\\
	=&(d\theta_{t_{0,j},p}|_{z=s}\cdot d_{x}\lambda(s))_{+}\p_{s}.
\end{align*}
The remaining cases follow similarly. 
\end{proof}

For $m=0$, this hierarchy is the dispersionless limit of the open Gelfand-Dickey hierarchy:
\begin{align*}
	\frac{\p L}{\p t_{p}}=[(L^{\frac{p}{n_{0}}}),L],\quad \frac{\p s}{\p t_{p}}=\p_{x}\cdot \rho ((L^{\frac{p}{n_{0}}})_{+})(1),\quad p=1,2,\cdots
\end{align*}
where $L=\p_{x}^{n_{0}}+\sum_{j=0}^{n_{0}-2}a_{j}(x)\p_{x}^{j}$, and 
$$
\rho: \sum_{j}a_{j}(x)\p^{j}_{x}\to \sum_{j}a_{j}(x)(\p_{x}+s(x))^{j}
$$
is an automorphism of the pseudo-differential operator algebra $\mathcal{A}$ \cite{wu2016extension}. 

The commutativity of the open Gelfand-Dickey hierarchy can be shown as follows. Denote
$$
\res A dx=b_{-1}
$$
for $A=\sum_{j}b_{j}\p^{j}\in\mathcal{A}$, then we have
$$
\res \rho(A) dx=\res A dx.
$$
From the equality
\begin{align*}
	\rho ((L^{\frac{p}{n_{0}}})_{+})(1)=&\res \rho ((L^{\frac{p}{n_{0}}})_{+})\cdot\p_{x}^{-1}dx\\
	=&\res (L^{\frac{p}{2}})_{+}\cdot(\p_{x}-s(x))^{-1}dx
\end{align*}
and the zero curvature equation:
$$
\frac{\p (L^{\frac{p}{n_{0}}})_{+}}{\p t_{q}}-\frac{\p (L^{\frac{q}{n_{0}}})_{+}}{\p t_{p}}=[(L^{\frac{q}{n_{0}}})_{+},(L^{\frac{p}{n_{0}}})_{+}],
$$
we obtain
\begin{align*}
	\frac{\p^2 s}{\p t_{p}\p t_{q}}-\frac{\p^2 s}{\p t_{q}\p t_{p}}&=\p_{x}\cdot\rho([(L^{\frac{p}{n_{0}}})_{+},(L^{\frac{q}{n_{0}}})_{+}])(1)-\p_{x}\cdot (L^{\frac{p}{n_{0}}})_{+}\cdot (L^{\frac{q}{n_{0}}})_{+}(1)+\p_{x}\cdot (L^{\frac{q}{n_{0}}})_{+}\cdot (L^{\frac{p}{n_{0}}})_{+}(1)\\
	&=0.
\end{align*}

Let $s(x)=\frac{\p v(x)}{\p x}$, 
and by applying the identity:
$$
e^{-v(x)}\cdot\p_{x}\cdot e^{v(x)}=\p_{x}+s(x),
$$
we obtain
\begin{align*}
	\frac{\p s}{\p t_{p}}=\p_{x}\cdot e^{-v(x)}\cdot (L^{\frac{p}{n_{0}}}_{+})(e^{v(x)}).
\end{align*}
Thus,
\begin{align*}
	\frac{\p e^{v(x)}}{\p t_{p}}= (L^{\frac{p}{n_{0}}}_{+})(e^{v(x)}).
\end{align*}
This version of the open Gelfand-Dickey hierarchy appeared in \cite{buryak2024open}.

\section{Frobenius manifold with trigonometric superpotential}
\subsection{Definition of $M^{Toda}$}
Given positive integers \( n_0, \ldots, n_m \), let \(  M^{Toda} \) be the space of functions
\[
\lambda(\varphi) = \frac{1}{n_0}e^{n_0\varphi} + a_{0,n_0-1}e^{(n_0-1)\varphi} + \cdots + a_{0,0} + \sum_{i=1}^{m}\sum_{j=1}^{n_i}a_{i,j}(e^{\varphi}-a_{i,0})^{-j},
\]
where \( a_{1,0} = 0 \). For any \( \p', \p'', \p''' \in T_{\lambda(z)}M^{Toda} \), the flat metric on \(  M^{Toda} \) is defined as: 
\[
\langle \p', \p'' \rangle_{\eta} := \eta(\p', \p'') = \sum_{|\lambda|<\infty} \res_{d\lambda=0} \frac{\p'(\lambda(z)dz) \p''(\lambda(z)dz)}{z^{2}d\lambda(z)},
\]
and the \( (0,3) \)-type tensor is given by:
\[
c(\p', \p'', \p''') := \sum_{|\lambda|<\infty} \res_{d\lambda=0} \frac{\p'(\lambda(z)dz) \p''(\lambda(z)dz) \p'''(\lambda(z)dz)}{z^{2}d\lambda(z)dz},
\]
where \( z = e^{\varphi} \).
The equality
\[
c(\p', \p'', \p''') = \eta(\p' \circ \p'', \p''')
\]
defines the multiplication structure on  \( T_{\lambda(z)}M^{Toda} \). Introduce the vector fields \( e \) and \( E \) on \( M^{Toda} \), such that
\[
Lie_{e}\lambda(z) = 1, \quad Lie_{E}\lambda(z) = \lambda(z) - \frac{z}{n_{0}}\lambda'(z).
\]
The data set \( (M^{Toda}, \eta, \circ, e, E) \) forms a semisimple Frobenius manifold with charge \( d = 1 \). 

The flat coordinate system for the metric \( \eta \), denoted as
\[
\mathbf{t} = \{t_{0,j}\}_{j=1}^{n_{0}-1} \cup \{t_{1,j}\}_{j=0}^{n_{1}} \cup \cdots \cup \{t_{m,j}\}_{j=0}^{n_{m}},
\]
is determined by the following expansion for $\varphi$:
\begin{equation}\label{todaflatc}
	\varphi=\left\{
	\begin{aligned}
		&t_{i,0}+t_{i,1}w_{i}^{-1}+\cdots,\quad &e^{\varphi}&\to a_{i,0},\ i=2,\cdots,m,\\
		&-\log(w_{1})+t_{1,0}+t_{1,1}w_{1}^{-1}+\cdots,\quad &e^{\varphi}&\to 0,\\
		&\log(w_{0})-t_{0,1}w_{0}^{-1}-t_{0,2}w_{0}^{-2}-\cdots,\quad &e^{\varphi}&\to \infty,
	\end{aligned}
	\right.
\end{equation}
where
	\begin{equation}\label{todasq}
	w_{i}=\left\{
	\begin{aligned}
		&(n_{i}\lambda)^{\frac{1}{n_{i}}}=w_{i,1}(z-a_{i,0})^{-1}+\cdots,\quad &z&\to a_{i,0},\ i=1,\cdots,m,\\
		&(n_{0}\lambda)^{\frac{1}{n_{0}}}=z+w_{0,0}+w_{0,1}z^{-1}+\cdots,\quad &z&\to \infty,\ i=0.\\
	\end{aligned}
	\right.
\end{equation}
Furthermore, we have
\begin{equation*}
	z^{-1}\p_{t_{i,j}}\lambda(z)=\left\{
	\begin{aligned}
		&-(w_{i}(z)^{n_{i}-j-1}w_{i}'(z))_{a_{i,0},\le -1},\quad i=1,\cdots,m,\ j=0,\cdots,n_{i},\\
		&(w_{i}(z)^{n_{0}-j-1}w_{0}'(z))_{\infty,\ge 0},\quad i=0,\ j=1,\cdots,n_{0}-1.
	\end{aligned}
	\right.
\end{equation*}
In this flat coordinate system, the vector fields \(e\) and \( E \) can be expressed as
	$$e=\frac{\p}{\p t_{1,n_{1}}}$$
	and
	$$
E=\sum_{j=1}^{n_{0}-1}\frac{j}{n_{0}}t_{0,j}\frac{\p}{\p t_{0,j}}+(\frac{1}{n_{0}}+\frac{1}{n_{1}})\frac{\p}{\p t_{1,0}}+\sum_{i=1}^{m}\sum_{j=1}^{n_{1}}\frac{j}{n_{i}}t_{i,j}\frac{\p}{\p t_{i,j}}+\sum_{i=2}^{m}\frac{1}{n_{0}}\frac{\p}{\p t_{i,0}}.
$$
\begin{lem}\label{todaconexp}
	Let \( \nabla \) be the Levi-Civita connection associated with the metric \( \eta \). Then, for any vector fields \( \p_1 \) and \( \p_2 \) on \( M^{Toda} \), we have
\begin{equation}\label{todaconn}
	(\nabla_{\p_{1}}\p_{2})\cdot\lambda(z)=\p_{1}\p_{2}\lambda(z)-z(\frac{\p_{1}\lambda(z)\p_{2}\lambda(z)}{z\lambda'(z)})'_{\infty,\ge 0}-\sum_{s=1}^{m}z(\frac{\p_{1}\lambda(z)\p_{2}\lambda(z)}{z\lambda'(z)})'_{a_{s,0},\le -1}.
\end{equation}
\end{lem}
\begin{proof}
	The proof follows the approach of Lemma \ref{kdvconexp}.
\end{proof}

\subsection{cotangent space and Hamiltionian structure}
For any point \( \lambda(z) \) in \( M^{Toda} \), a tangent vector \( \p \in T_{\lambda(z)}M^{Toda} \) can be represented as \( \xi(z) = \p \lambda(z) \), where
\[
\xi(z) = b_{0,n_{0}-1}z^{n_{0}-1} + \cdots + b_{0,0} + \sum_{j=1}^{n_{1}} b_{1,j}z^{-j} + \sum_{i=2}^{m} \sum_{j=1}^{n_{i}+1} b_{i,j}(z-a_{i,0})^{-j}.
\]

To describe a cotangent vector at \( \lambda(z) \), we follow the procedure in Section \ref{seccot}. Consider disjoint disks \( D_1, \ldots, D_m \) in the complex plane with $a_{i,0}\in D_{i}$, and let \( \gamma_i = \partial D_i \). Define the space \( \mathcal{H} \) of analytic function germs on \( \cup_{s=1}^{m} \gamma_s \), and introduce the following pairing:
\[
\langle \omega(z), \xi(z) \rangle := \frac{1}{2\pi\mathrm{i}} \sum_{s=1}^{m} \int_{\gamma_s} \omega(z) \xi(z) \frac{dz}{z}, \quad \omega(z) \in \mathcal{H}, \quad \xi(z) \in T_{\lambda(z)}M^{Toda}.
\]
This pairing induces a surjective map from \( \mathcal{H} \) to \( T^{\ast}_{\lambda(z)}M^{Toda} \), allowing elements of \( \mathcal{H} \) to be regarded as cotangent vectors at \( \lambda(z) \).

\begin{lem}
	For any \( \xi(z) \in T_{\lambda(z)}M^{Toda} \), define the linear maps from \( \mathcal{H} \) to \( T_{\lambda(z)}M^{Toda} \) as
	\begin{equation}\label{todadulmet}
		\eta^{\ast}(\omega(z)) = -(\omega(z))_{+}\lambda'(z)z + (\omega(z)\lambda'(z))_{+}z,
	\end{equation}
	and
	\begin{equation}\label{todadulpro}
		C_{\xi(z)}(\omega(z)) = -(\omega(z)\xi(z))_{+}\lambda'(z)z + (\omega(z)\lambda'(z))_{+}z\xi(z).
	\end{equation}
	We then obtain
	\[
	\langle \eta^{\ast}(\omega(z)), \xi(z) \rangle_{\eta} = \langle \omega(z), \xi(z) \rangle,
	\]
	and
	\[
	c(\xi_{1}(z), \xi_{2}(z), \eta^{\ast}(\omega(z))) = \langle C_{\xi_{1}}(\omega(z)), \xi_{2}(z) \rangle_{\eta}.
	\]
\end{lem}

\begin{proof}
	The proof is analogous to the argument presented in Lemma \ref{kdvdulmap}.
\end{proof}

Using the equality
\[
\mathcal{P}(\omega) = \eta^{\ast} \cdot \nabla_{\p_{x}} \omega,
\]
we can derive the explicit form of the Hamiltonian structure \( \mathcal{P} \) associated with the flat metric \( \langle\ ,\ \rangle_{\eta} \).

\begin{cor}
	The dispersionless Hamiltonian operator \( \mathcal{P} \) associated with the metric \( \langle\ ,\ \rangle_{\eta} \) has the form
	\begin{equation}\label{todaham}
		\mathcal{P}(\omega(z)) = \{\omega(z)_{+}, \lambda(z)\}z - \{\omega(z), \lambda(z)\}_{+}z,
	\end{equation}
	where
	\[
	\{f(z,x), g(z,x)\} = \frac{\partial f(z,x)}{\partial z} \frac{\partial g(z,x)}{\partial x} - \frac{\partial f(z,x)}{\partial x} \frac{\partial g(z,x)}{\partial z}.
	\]
\end{cor}

\begin{proof}
	The proof follows the same approach as in the proof of Lemma \ref{kdvhamthm}.
\end{proof}

\subsection{principal hierarchy for $M^{Toda}$}
To prove Theorem \ref{maincKP}, we follow a similar approach to that in Section \ref{kdvpriproof}, requiring the following lemma:
\begin{lem}
	For any $p \in \mathbb{N}$, let \( Q_{p}(\lambda)  \) be analytic functions in \( \lambda \) that satisfy
	\[
	\frac{\partial Q_{p}(\lambda)}{\partial \lambda} = Q_{p-1}(\lambda).
	\]
	Define
	\begin{equation}\label{todaFdef}
		F_{i,p} = \frac{1}{2\pi i} \int_{\gamma_{i}} Q_{p+1}(\lambda(z)) \, \frac{dz}{z},
	\end{equation}
	then
	\begin{equation}\label{Fcond2}
		\eta^{\ast} \cdot \nabla_{\partial} dF_{i,p} = C_{\partial}(dF_{i,p-1}),
	\end{equation}
	where the operators \( \eta^{\ast} \cdot \nabla_{\partial} \) and \( C_{\partial} \) are defined by \eqref{todadulmet} and \eqref{todadulpro}, respectively.
\end{lem}

\begin{proof}[Proof of Theorem \ref{mainToda}]
The proof follows a similar approach to that of Theorem \ref{maincKP}. As an instance, for the density \( \theta_{t_{i,n_{i}},p} \),
let \( M' \subset M^{Toda} \) be such that \( w_{0} = \lambda^{\frac{1}{n_{0}}} \) can be analytically continued onto \( \cup_{s=1}^{m} \gamma_{s} \), with the winding number \( 1 \) along \( \gamma_{i} \) and \( 0 \) along \( \gamma_{j} \) for \( j \neq i \). Similarly, \( w_{i} = \lambda^{\frac{1}{n_{i}}} \) can be analytically continued onto \( \gamma_{i} \) with the winding number \( -1 \). Then, on \( M' \), we have
\begin{align*}
	\theta_{t_{i,n_{i}},p} =& \res_{\infty} \frac{c_{p}}{n_{0}} \frac{\lambda^{p}}{p!} \frac{dz}{z} + \frac{1}{2\pi\mathrm{i}} \int_{\gamma_{i}} \frac{\lambda^{p}}{p!} (\log w_{0}w_{i} - \frac{c_{p}}{n_{i}}) \frac{dz}{z} + \frac{1}{2\pi\mathrm{i}} \sum_{s \neq i} \int_{\gamma_{s}} \frac{\lambda^{p}}{p!} \log w_{0} \frac{dz}{z},\\
	=& \frac{1}{2\pi\mathrm{i}} \int_{\gamma_{i}} \frac{\lambda^{p}}{p!} (\log w_{0}w_{i} - \frac{c_{p}}{n_{i}} - \frac{c_{p}}{n_{0}}) \frac{dz}{z} + \frac{1}{2\pi\mathrm{i}} \sum_{s \neq i} \int_{\gamma_{s}} \frac{\lambda^{p}}{p!} (\log w_{0} - \frac{c_{p}}{n_{0}}) \frac{dz}{z},
\end{align*}
which satisfies equality \eqref{princon1}.

 Introduce the operator \( \mathcal{E} = E + \frac{1}{n_{0}}z\frac{\partial}{\partial z} \), then
\[
\operatorname{Lie}_{\mathcal{E}}\lambda(z) = \lambda(z),
\]
and
\begin{equation}\label{eulerlem2}
	\res_{a_{i,0}} \operatorname{Lie}_{\mathcal{E}} f(\lambda(z)) \frac{dz}{z} = \operatorname{Lie}_E \res_{a_{i,0}} f(\lambda(z)) \frac{dz}{z}.
\end{equation}
Thus, we obtain
\[
\operatorname{Lie}_E \theta_{u, p}(t) = \begin{cases}
	\left(p + 1 - \frac{j}{n_{0}}\right) \theta_{u, p}(t), & u = t_{0,j}; \\
	\left(p + 1 - \frac{j}{n_{i}}\right) \theta_{u, p}(t), & u = t_{i,j},\ j \neq n_{i}; \\
	p \theta_{u, p}(t) + \sum_{s=1}^{m} \frac{1}{n_{0}} \theta_{t_{s,0},p-1} + \frac{1}{n_{i}} \theta_{t_{i,0},p-1}, & u = t_{i,n_{i}}.
\end{cases}
\]
Therefore, equality \eqref{princon2} holds.

Finally, using formula \eqref{todaham}, we deduce the Hamiltonian vector fields corresponding to the densities \( \theta_{\alpha,p} \).

This completes the proof of the theorem.
	\end{proof}
	\subsection{principal hierarchy for $M^{C-Toda}$}
Let \( M^{C-Toda} \) be the submanifold of \( M^{Toda} \) consisting of functions of the form \eqref{Ctodasup}. We will first show that \( M^{C-Toda} \) is a natural Frobenius submanifold of \( M^{Toda} \).

For any \( \lambda(z) \in M^{C-Toda} \), we have:

	\begin{align*}
		&w_{0}(\frac{1}{p})=w_{1}(p),\quad p\to 0;\\
		&w_{2}(\frac{1}{p})=-w_{2}(p),\quad p\to 1;\\
		&w_{3}(\frac{1}{p})=-w_{3}(p),\quad p\to -1;\\
		&w_{2i-2}(\frac{1}{p})=w_{2i-1}(p),\quad p\to b_{2i-1,0},\ i=3,4,\cdots,m',
	\end{align*}
	 where $b_{2i-1,0}+\frac{1}{b_{2i-1,0}}=\tilde{b}_{i,0}$. For $p\to 0$, we have
\begin{align*}
	-\tilde{\varphi}=&\log w_{0}(\frac{1}{p})-\frac{1}{2}t_{1,0}-t_{0,1}w_{0}^{-1}(\frac{1}{p})-\cdots\\
	=&\log w_{1}(p)-\frac{1}{2}t_{1,0}-t_{0,1}w_{1}^{-1}(p)-\cdots\\
	=&\log w_{1}(p)-\frac{1}{2}t_{1,0}-t_{1,1}w_{1}^{-1}(p)-\cdots.
\end{align*}
	For $p\to 1$:
\begin{align*}
	-\tilde{\varphi}=&t_{2,0}-\frac{1}{2}t_{1,0}+w_{3}^{-1}(\frac{1}{p})+\cdots\\
	=&-\log 1+t_{2,1}w_{2}^{-1}(\frac{1}{p})+\cdots\\
	=&-\log 1-t_{2,1} w_{2}^{-1}(p)+t_{2,2} w_{2}^{-2}(p)-\cdots\\
	=&-\log 1-t_{2,1} w_{2}^{-1}(p)-t_{2,2} w_{2}^{-2}(p)-\cdots.
\end{align*}
  For $p\to -1$:
  \begin{align*}
  -\tilde{\varphi}=&t_{3,0}-\frac{1}{2}t_{1,0}+w_{3}^{-1}(\frac{1}{p})+\cdots\\
  =&-\log (-1)+t_{3,1}w_{3}^{-1}(\frac{1}{p})+\cdots\\
  =&-\log (-1)-t_{3,1} w_{3}^{-1}(p)+t_{3,2} w_{3}^{-2}(p)-\cdots\\
  =&-\log (-1)-t_{3,1} w_{3}^{-1}(p)-t_{3,2} w_{3}^{-2}(p)-\cdots.
  \end{align*}
	For $p\to b_{2i-1}=\frac{1}{b_{2i-2}}$:
\begin{align*}
	-\tilde{\varphi}=&t_{2i-2,0}-\frac{1}{2}t_{1,0}+t_{2i-2,1}w_{2i-2}^{-1}(\frac{1}{p})+\cdots\\
	=&t_{2i-2,0}-\frac{1}{2}t_{1,0}+t_{2i-2,1}w_{2i-1}^{-1}(p)+\cdots\\
	=&t_{2i-1,0}-\frac{1}{2}t_{1,0}+t_{2i-1,1}w_{2i-1}^{-1}(p)+\cdots.
\end{align*}
	Thus, the flat coordinates satisfy the following constraints:
\begin{align*}
	&t_{0,j}=t_{1,j},\quad j=1,\cdots,n_{0}-1;\\
	&t_{2,2}=t_{2,4}=\cdots= t_{2,n_{2}}=0,\quad t_{2,0}-\frac{1}{2}t_{1,0}=-\log 1;\\
	&t_{3,2}=t_{3,4}=\cdots= t_{3,n_{3}}=0,\quad t_{3,0}-\frac{1}{2}t_{1,0}=-\log (-1);\\
	&t_{2i-2,j}=t_{2i-1,j},\quad i=3,\cdots,m',\ j=0,\cdots,n_{2i-2}.
\end{align*}
Here, the values of \( \log 1 \) and \( \log(-1) \) depend on the chosen branch of \( \log p \). Hence, \( M^{C-Toda} \) forms a flat submanifold of \( M^{Toda} \) with flat coordinates:
\[
\mathbf{t} = \{t_{1,j}\}_{j=0}^{n_0'} \cup \{t_{2,2j-1}\}_{j=1}^{n_1'} \cup \{t_{3,2j-1}\}_{j=1}^{n_2'} \cup \{t_{4,j}\}_{j=0}^{n_3'} \cup \cdots \cup \{t_{2m'-2,j}\}_{j=0}^{n_{m'}'}.
\]

	On the other hand, let $p_{1},\cdots,p_{r}$ be the simple critical points of $\lambda(z)\in M^{Toda}$, where $r=dim(M^{Toda})$. The critical values $u_{j}=\lambda(p_{j})$ form the canonical coordinates for $M^{Toda}$. For $\lambda(z)\in M^{C-Toda}$, we can choose $$
	p_{2j}=\frac{1}{p_{2j-1}},\quad j=1,2,\cdots,\frac{r}{2},
	$$
	which implies
	$$
	u_{2j}=u_{2j-1},\quad j=1,2,\cdots,\frac{r}{2},
	$$
	Thus, $M^{C-Toda}$ is a caustic submanifold of $M^{Toda}$. According to Corollary 3.7 in \cite{STRACHAN2004}, $M^{C-Toda}$ is a natural Frobenius submanifold of $M^{Toda}$.
	
	\begin{rem}
		Note that the above conclusion also holds when \( n_1' = 0 \) or \( n_2' = 0 \). In the case \( n_1' = 0 \), we have
		$$
		r = \dim(M^{Toda}) = 2\dim(M^{C-Toda}) - 1.
		$$
		Choose simple critical points \( p_1, \ldots, p_r \) such that
		$$
		p_1 = 1, \quad p_{2j} = \frac{1}{p_{2j+1}}, \quad j = 1, 2, \ldots. \frac{r-1}{2},
		$$
		This implies
		$$
		u_{2j} = u_{2j+1}, \quad j = 1, 2, \ldots, \frac{r-1}{2}.
		$$
	\end{rem}
	
Let us now consider the principal hierarchy for \( M^{C-Toda} \). Define \( \mathcal{H}^{odd} \) as the subspace of \( \mathcal{H} \) consisting of elements satisfying \( \omega(z^{-1}) = -\omega(z) \). For \( \lambda(z) \in M^{C-Toda} \) and \( \omega(z) \in \mathcal{H}^{odd} \), the following relations hold:
\[
\omega_{-}(z^{-1}) = \begin{cases}
	-\omega_{-}(z), & \text{if } z, z^{-1} \in D^{c}, \\
	\omega_{+}(z), & \text{if } z^{-1} \in D^{c}, z \in D,
\end{cases}
\]
and
\[
\omega_{+}(z^{-1}) = \begin{cases}
	\omega_{+}(z), & \text{if } z, z^{-1} \in D, \\
	-\omega_{-}(z), & \text{if } z^{-1} \in D, z \in D^{c}.
\end{cases}
\]
Therefore, \( \eta^{\ast}(\omega(z)) \in T_{\lambda(z)}M^{C-Toda} \), and for any \( \p \in T_{\lambda(z)}M^{C-Toda} \), it follows that:
\[
\eta^{\ast} \cdot \nabla_{\p} \omega \in T_{\lambda(z)}M^{C-Toda}.
\]
Since the differentials of the Hamiltonian densities
\[
\begin{aligned}
	&\theta_{t_{0,j},p} + \theta_{t_{1,j},p}, \quad j=1,2,\dots,n_{0}'; \\
	&\theta_{t_{2,2j-1},p}, \quad j=1,2,\dots,n_{1}'; \\
	&\theta_{t_{3,2j-1},p}, \quad j=1,2,\dots,n_{2}'; \\
	&\theta_{t_{2i-2,j},p} + \theta_{t_{2i-1,j},p}, \quad i=3,\dots,m', \ j=0,\dots,n_{i}' 
\end{aligned}
\]
belong to \( \mathcal{H}^{odd} \), the corresponding Hamiltonian vector fields can be restricted to the loop space of \( M^{C-Toda} \), thereby forming the principal hierarchy for \( M^{C-Toda} \).
	
	\subsection{principal hierarchy for $\hat{M}^{Toda}$}
Let \( \hat{M}^{Toda} \) denote the almost duality of the Frobenius manifold \( M^{Toda} \). In this subsection, we will construct the principal hierarchy for \( \hat{M}^{Toda} \). 

Suppose $I\in \textbf{D}$, where \( I \) is the set of all zeros and poles of \( \lambda(z) \) in  \( \mathbb{C} \). Define the linear map from \( \mathcal{H} \) to \( T_{\lambda(z)}M^{Toda} \) as \( g^{\ast}(\omega(z)) = C_{E}(\omega(z)) \), that is,
\[
g^{\ast}(\omega(z)) = (\lambda(z)\omega(z))_{-}\lambda'(z)z - z\lambda(z)(\lambda'(z)\omega(z))_{-} + \frac{z\lambda'(z)}{n_{0}}\frac{1}{2\pi\mathrm{i}} \sum_{s=1}^{m} \int_{\gamma_{s}} \lambda'(z)\omega(z) \, dz,
\]
then we have
\[
\langle \omega(z), \xi(z) \rangle = (g^{\ast}(\omega(z)), \xi(z))_{g},
\]
where
\[
(\p_{1}, \p_{2})_{g} = \sum \res_{d\lambda=0} \frac{\p_{1}\log\lambda(z) \p_{2}\log\lambda(z)}{(\log\lambda(z))'z^{2}} \, dz.
\]

Let \( \hat{\nabla} \) be the Levi-Civita connection associated with the intersection form \( g \), we have
\[
\hat{\nabla}_{\p_{1}}\p_{2} \cdot \lambda(z) = \p_{1}\p_{2}\lambda(z) - \frac{\p_{1}\lambda(z)\p_{2}\lambda(z)}{\lambda(z)} - z\lambda(z) \sum_{q \in I} \left( \frac{\p_{1}\lambda(z)\p_{2}\lambda(z)}{z\lambda(z)\lambda'(z)} \right)_{q,\le-1}'.
\] Furthermore, 
\[
g^{\ast} \cdot \hat{\nabla}_{\p}\omega(z) = \{\omega(z),\lambda(z)\}_{-}\lambda(z)z - \{(\omega(z)\lambda(z))_{-},\lambda(z)\}z - \frac{z\lambda'(z)}{n_{0}}\frac{1}{2\pi\mathrm{i}} \sum_{s=1}^{m} \int_{\gamma_{s}} \{\omega(z),\lambda(z)\} \, dz,
\]
where \( \{f,g\} = f'\p g - g'\p f \). Setting \( \p = \p_{x} \), we obtain the explicit form of the Hamiltonian structure $\hat{\mathcal{P}}$ for the flat metric $g$.
\begin{proof}[Proof of Theorem \ref{mainTodadul}]
	Consider functions \( F_{p} = \frac{1}{2\pi \mathrm{i}} \int_{\gamma_{j}} \tilde{Q}_{p}(\lambda) \frac{dz}{z} \) on \( M^{Toda
	} \), where \( Q_{p}(\lambda) = \frac{\partial \tilde{Q}_{p}(\lambda)}{\partial \lambda} \) satisfy the recurrence relation:
\begin{equation}\label{todadulrec}
		\lambda \frac{\partial Q_{p+1}}{\partial \lambda} + Q_{p+1} = Q_{p}.
\end{equation}
	We have \( dF_{p} = Q_{p} \mathbf{1}_{\gamma_{j}} \) and
	\begin{align*}
		g^{\ast} \cdot \hat{\nabla}_{\p} dF_{p+1} =& - \left( \left( \frac{\partial Q_{p+1}}{\partial \lambda} \lambda \lambda' \mathbf{1}_{\gamma_{j}} + Q_{p+1} \lambda' \mathbf{1}_{\gamma_{j}} \right)_{-} \p \lambda z \right) \\
		&+ \left( \left( \frac{\partial Q_{p+1}}{\partial \lambda} \lambda \p \lambda \mathbf{1}_{\gamma_{j}} + Q_{p+1} \p \lambda \mathbf{1}_{\gamma_{j}} \right)_{-} \lambda' z \right) \\
		=& \left( (Q_{p} \p \lambda \mathbf{1}_{\gamma_{j}})_{-} \lambda' z \right) - \left( (Q_{p} \lambda' \mathbf{1}_{\gamma_{j}})_{-} \p \lambda z \right) \\
		=& C_{\p}(dF_{p}).
	\end{align*}
	
	Let \( M' \) be a subspace of \( M^{Toda} \) such that for any \( \lambda(z) \in M' \), the winding number of \( \lambda(z) \) along \( \gamma_{j} \) is zero. Define \( \tilde{Q}(p) = \frac{(\log(\lambda(z)))^{p+1}}{(p+1)!} \), then \( F_{\gamma_{j},p} \) satisfies equality \eqref{todadulrec} on \( M' \). Furthermore, by the uniqueness of analytic function, for any \( \tilde{\gamma} \in [\gamma_{j}] \), \( F_{\tilde{\gamma},p} \) satisfies equality \eqref{princon1}, where \( [\gamma_{j}] \) denotes the homotopy equivalence class of \( \gamma_{j} \) in \( \mathbb{C}-I \).
	
	In particular, let \( q_{1}, \ldots, q_{k}, p_{1}, \ldots, p_{k} \) be the zeros and poles of \( \lambda(z) \) within the region surrounded by \( \tilde{\gamma} \). Additionally, let $q_{k+1}, \ldots, q_{n}, p_{k+1}, \ldots, p_{r}$ be the zeros and poles outside this region. If \( 0 \in \mathbb{C} \) is outside the region surrounded by \( \tilde{\gamma} \), we have
	\begin{align*}
		F_{\tilde{\gamma},0} =& \frac{1}{2\pi\mathrm{i}} \int_{\tilde{\gamma}} \log(\lambda(z)) \frac{dz}{z} \\
		=& \frac{1}{2\pi\mathrm{i}} \int_{\tilde{\gamma}} (\log(\lambda(z) \prod_{s=1}^{k} \frac{z-p_{s}}{z-q_{s}}) + \log(\prod_{s=1}^{k} \frac{z-q_{s}}{z-p_{s}}) )\frac{dz}{z} \\
		=& -\res_{0} \log(\prod_{s=1}^{k} \frac{z-q_{s}}{z-p_{s}}) \, \frac{dz}{z} \\
		=& \sum_{s=1}^{k} \log(p_{s}) - \sum_{s=1}^{k} \log(q_{s}).
	\end{align*}
	Otherwise
	\begin{align*}
		F_{\tilde{\gamma},0} =& \frac{1}{2\pi\mathrm{i}} \int_{\tilde{\gamma}} \log(\lambda(z)) \frac{dz}{z} \\
		=& \frac{1}{2\pi\mathrm{i}} \int_{\tilde{\gamma}} (\log(\lambda(z) \prod_{s=1}^{k} \frac{z-p_{s}}{z-q_{s}}) + \log(\prod_{s=1}^{k} \frac{z-q_{s}}{z-p_{s}}) )\frac{dz}{z} \\
		=& \res_{0} \log(\lambda(z) \prod_{s=1}^{k} \frac{z-p_{s}}{z-q_{s}}) \frac{dz}{z} \\
		=& \sum_{s=k+1}^{r} \log(p_{s}) - \sum_{s=k+1}^{n} \log(q_{s}).
	\end{align*}
	The theorem is proved.
\end{proof}
	\begin{cor}
		The Hamiltonian vector fields \(\frac{\partial}{\partial\hat{T}^{\gamma_{i},p}} = \hat{\mathcal{P}}(d F_{\gamma_{i},p+1})\) take the form:
		\[
		\frac{\partial \lambda(z)}{\partial\hat{T}^{\gamma_{i},p-1}} = \{(\frac{(\log(\lambda(z)))^{p}}{p!} \mathbf{1}_{\gamma_{i}})_{+}, \lambda(z)\},
		\]
		where
		\begin{equation}
			\{f(z,x),g(z,x)\} = z \left( \frac{\partial f(z,x)}{\partial z} \frac{\partial g(z,x)}{\partial x} - \frac{\partial f(z,x)}{\partial x} \frac{\partial g(z,x)}{\partial z} \right).
		\end{equation}
	\end{cor}
	\subsection{rank-1 extension}
	
By using a similar approach as in subsection \ref{extcKP}, we can construct a rank-1 extension of $M^{Toda}$ using the following lemma.
	\begin{lem}\label{todaopenlem}
		For any vector field $\p_{1},\p_{2}$ on $M^{Toda}$, we have
		$$
		\frac{\p_{1}\lambda(z)\p_{2}\lambda(z)-\p_{1}\lambda(z)\circ \p_{2}\lambda(z)}{z\lambda'}=(\frac{\p_{1}\lambda(z)\p_{2}\lambda(z)}{z\lambda'(z)})_{\infty,\ge 0}+\sum_{j=1}^{m}(\frac{\p_{1}\lambda(z)\p_{2}\lambda(z)}{z\lambda'(z)})_{\varphi_{j},\le -1}.
		$$
	\end{lem}
	
	\begin{cor}
		Let $z=e^{s}$ and define $\Omega(\textbf{t},s)$ such that
		$$
		\p_{s}\p_{\alpha}\Omega=\p_{\alpha}\lambda(z),\quad \p^{2}_{s}\Omega=z\lambda'(z),\quad \p_{\alpha}\p_{\beta}\Omega=(\frac{\p_{\alpha}\lambda(z)\p_{\beta}\lambda(z)}{z\lambda'(z)})_{\infty,\ge 0}+\sum_{j=1}^{m}(\frac{\p_{\alpha}\lambda(z)\p_{\beta}\lambda(z)}{z\lambda'(z)})_{\varphi_{j},\le -1}.
		$$
		Then $\omega=\p_{s}\Omega$ satisfies the condition of Lemma \ref{consopen}, thus defining a flat F-manifold structure on $M^{Toda}\times\mathbb{C}$, with multiplication given by the following expressions:
		\begin{align*}
			&(\p_{\alpha},0)\star(\p_{\beta},0)=(\p_{\alpha}\circ\p_{\beta},\p_{\alpha}\p_{\beta}\Omega\cdot\p_{s}),\\
			&(\p_{\alpha},0)\star(0,\p_{s})=(0,\p_{\alpha}\lambda(z)\cdot\p_{s}),\\
			&(0,\p_{s})\star(0,\p_{s})=(0,z\lambda'(z)\cdot\p_{s}).
		\end{align*}
	\end{cor}
	
		The principal hierarchy for the flat F-manifold $M^{Toda}\times \mathbb{C}$ is given by
	$$
	\frac{\p}{\p\tilde{T}^{\bullet,p}}=\Theta_{\bullet,p}\star\tilde{\p}_{x},
	$$
	where
	$$
	\tilde{\p}_{x}=(\p_{x},\frac{\p s}{\p x}\p_{s}),
	$$
    and the vector fields $\Theta_{\bullet,p}$ are given by the following theorem.
	\begin{thm}
		The vector fields on $M^{Toda}\times\mathbb{C}$ of the form
		\begin{align*}
			&\Theta_{s,p}=(0,\frac{\lambda(z)^{p}}{p!}\p_{s})\\
			&\Theta_{t_{0,n_{0}-j},p}	=(\eta^{\ast}(d\theta_{t_{0,j},p}),(d\theta_{t_{0,j},p}|_{z=e^{s}})_{+}\p_{s}),\\
			&\Theta_{t_{i,n_{i}-j},p}=(\eta^{\ast}(d\theta_{t_{i,j},p}),-(d\theta_{t_{i,j},p}|_{z=e^{s}})_{-}\p_{s}),\quad i=1,\cdots,m,
		\end{align*}
		satisfy equations of the form \eqref{openreu1} and \eqref{openreu2}. 
	\end{thm}

By applying a similar approach to the proof of Corollary \ref{cKPopenpri}, we derive the explicit form of the principal hierarchy for the flat F-manifold $M^{Toda}\times \mathbb{C}$ as follows:
		\begin{align*}
			&\frac{\p}{\p \tilde{T}^{\bullet,p}}=(0,d_{x}\lambda(z)\cdot\frac{\lambda(z)^{p}}{p!});\\
			&\frac{\p}{\p\tilde{T}^{t_{0,n_{0}-j},p}}=(\frac{\p}{\p T^{t_{0,j},p}},(d\theta_{t_{0,j},p}|_{z=e^{s}}\cdot d_{x}\lambda(z))_{+}\p_{s});\\
			&\frac{\p}{\p\tilde{T}^{t_{i,n_{i}-j},p}}=(\frac{\p}{\p T^{t_{i,j},p}},-(d\theta_{t_{i,j},p}|_{z=e^{s}}\cdot d_{x}\lambda(z))_{-}\p_{s}),\quad i=1,\cdots,m,\\
		\end{align*}
		where
		$$
		d_{x}\lambda(z)=\p_{x}\lambda(z)+e^{s}\lambda'(z)\frac{\p s}{\p x}.
		$$

\bibliography{mybib}
\bibliographystyle{unsrt}

\end{document}